\PassOptionsToPackage{unicode}{hyperref}
\PassOptionsToPackage{hyphens}{url}
\PassOptionsToPackage{dvipsnames,svgnames,x11names}{xcolor}
\documentclass[
  12pt]{article}

\usepackage{amsmath,amssymb}
\usepackage{iftex}
\ifPDFTeX
  \usepackage[T1]{fontenc}
  \usepackage[utf8]{inputenc}
  \usepackage{textcomp} 
\else 
  \usepackage{unicode-math}
  \defaultfontfeatures{Scale=MatchLowercase}
  \defaultfontfeatures[\rmfamily]{Ligatures=TeX,Scale=1}
\fi
\usepackage{lmodern}
\ifPDFTeX\else  
\fi
\IfFileExists{upquote.sty}{\usepackage{upquote}}{}
\IfFileExists{microtype.sty}{
  \usepackage[]{microtype}
  \UseMicrotypeSet[protrusion]{basicmath} 
}{}
\makeatletter
\@ifundefined{KOMAClassName}{
  \IfFileExists{parskip.sty}{%
    \usepackage{parskip}
  }{
    \setlength{\parindent}{0pt}
    \setlength{\parskip}{6pt plus 2pt minus 1pt}}
}{
  \KOMAoptions{parskip=half}}
\makeatother
\usepackage{xcolor}
\makeatletter
\ifx\paragraph\undefined\else
  \let\oldparagraph\paragraph
  \renewcommand{\paragraph}{
    \@ifstar
      \xxxParagraphStar
      \xxxParagraphNoStar
  }
  \newcommand{\xxxParagraphStar}[1]{\oldparagraph*{#1}\mbox{}}
  \newcommand{\xxxParagraphNoStar}[1]{\oldparagraph{#1}\mbox{}}
\fi
\ifx\subparagraph\undefined\else
  \let\oldsubparagraph\subparagraph
  \renewcommand{\subparagraph}{
    \@ifstar
      \xxxSubParagraphStar
      \xxxSubParagraphNoStar
  }
  \newcommand{\xxxSubParagraphStar}[1]{\oldsubparagraph*{#1}\mbox{}}
  \newcommand{\xxxSubParagraphNoStar}[1]{\oldsubparagraph{#1}\mbox{}}
\fi
\makeatother

\usepackage{longtable,booktabs,array}
\usepackage{calc} 
\usepackage{etoolbox}
\makeatletter
\patchcmd\longtable{\par}{\if@noskipsec\mbox{}\fi\par}{}{}
\makeatother
\IfFileExists{footnotehyper.sty}{\usepackage{footnotehyper}}{\usepackage{footnote}}
\makesavenoteenv{longtable}
\usepackage{graphicx}
\makeatletter
\def\maxwidth{\ifdim\Gin@nat@width>\linewidth\linewidth\else\Gin@nat@width\fi}
\def\maxheight{\ifdim\Gin@nat@height>\textheight\textheight\else\Gin@nat@height\fi}
\makeatother
\setkeys{Gin}{width=\maxwidth,height=\maxheight,keepaspectratio}
\makeatletter
\def\fps@figure{htbp}
\makeatother

\makeatletter
\@ifpackageloaded{caption}{}{\usepackage{caption}}
\AtBeginDocument{%
\ifdefined\contentsname
  \renewcommand*\contentsname{Table of contents}
\else
  \newcommand\contentsname{Table of contents}
\fi
\ifdefined\listfigurename
  \renewcommand*\listfigurename{List of Figures}
\else
  \newcommand\listfigurename{List of Figures}
\fi
\ifdefined\listtablename
  \renewcommand*\listtablename{List of Tables}
\else
  \newcommand\listtablename{List of Tables}
\fi
\ifdefined\figurename
  \renewcommand*\figurename{Figure}
\else
  \newcommand\figurename{Figure}
\fi
\ifdefined\tablename
  \renewcommand*\tablename{Table}
\else
  \newcommand\tablename{Table}
\fi
}
\@ifpackageloaded{float}{}{\usepackage{float}}
\floatstyle{ruled}
\@ifundefined{c@chapter}{\newfloat{codelisting}{h}{lop}}{\newfloat{codelisting}{h}{lop}[chapter]}
\floatname{codelisting}{Listing}

\makeatother
\makeatletter
\@ifpackageloaded{caption}{}{\usepackage{caption}}
\@ifpackageloaded{subcaption}{}{\usepackage{subcaption}}
\makeatother

\ifLuaTeX
  \usepackage{selnolig}  
\fi
\usepackage[]{natbib}
\usepackage{bookmark}

\IfFileExists{xurl.sty}{\usepackage{xurl}}{} 
\hypersetup{
  pdftitle={Differentially Private Bayesian Inference for Gaussian Copula Correlations},
  pdfauthor={Shuo Wang; Joseph Feldman; Jerome P. Reiter},
  pdfkeywords={confidentiality; measurement; multivariate; privacy; probit.},
  colorlinks=true,
  linkcolor={blue},
  filecolor={Maroon},
  citecolor={Blue},
  urlcolor={Blue},
  pdfcreator={LaTeX via pandoc}}

\newcommand{\anon}{1}

\usepackage{amsthm}
\usepackage{amsfonts}
\usepackage{dsfont}
\usepackage{algorithm}
\usepackage{algorithmic}
\usepackage{comment}
\newtheorem{theorem}{Theorem}[section]

\newtheorem{proposition}[theorem]{Proposition}

\theoremstyle{definition}
\newtheorem{definition}[theorem]{Definition}
\theoremstyle{remark}

\newcommand{\cA}{\mathcal{A}}

\newcommand{\norm}[2]{\left\| #1 \right\|_{#2}}
\newcommand{\wt}{\widetilde}

\begin{document}

\def\spacingset#1{\renewcommand{\baselinestretch}%
{#1}\small\normalsize} \spacingset{1}


\if1\anon
{
  \title{\bf Differentially Private Bayesian Inference for Gaussian Copula Correlations}
  \author{Shuo Wang\thanks{
    The authors gratefully acknowledge NSF-SES-2217456.}\hspace{.2cm}\\
    Department of Statistical Science, Duke University\\
    and \\
    Joseph Feldman \\
    Department of Statistics and Data Science, Washington University in Saint Louis\\ 
    and \\ 
    Jerome P. Reiter \\ Department of Statistical Science, Duke University}
  \maketitle
} \fi

\if0\anon
{
  \bigskip
  \bigskip
  \bigskip
  \begin{center}
    {\LARGE\bf Differentially Private Bayesian Inference for Gaussian Copula Correlations}
\end{center}
  \medskip
} \fi

\bigskip
\begin{abstract}
Gaussian copulas are widely used to estimate multivariate distributions and relationships.  We present algorithms for estimating Gaussian copula correlations that ensure differential privacy. We first convert data values into sets of two-way tables of counts above and below marginal medians. We then add noise to these counts to satisfy differential privacy. We utilize the one-to-one correspondence between the true counts and the copula correlation to  estimate a posterior distribution of the copula correlation given the noisy counts, marginalizing over the distribution of the underlying true counts using a composite likelihood. 
We also present an alternative, maximum likelihood approach for point estimation. 
Using simulation studies, we compare these methods to extant methods in the literature for computing differentially private copula correlations.  
\end{abstract}

\noindent%
{\it Keywords:} confidentiality; measurement; multivariate; privacy; probit.
\vfill

\newpage

\newpage
\spacingset{1.8} 

\newpage

\section{Introduction}\label{sec:intro}

The Gaussian copula is used for data analyses across many disciplines including, for example, epidemiology \citep{feldman2025using}, finance \citep{pitt2006efficient}, health \citep{dobra}, marketing \citep{becker2022revisiting},  and psychology \citep{ince2017statistical}. It is also employed as an engine for imputation of missing data \citep[e.g., ][]{hollenbachetal,  zhao2020missing, mdgc,   feldman2025using}  and for generation of synthetic data \citep[e.g., ][]{jeong, li:zhao, benali, feldman2022bayesian}.
Its popularity for applied data analysis seemingly stems from its construction: it allows for arbitrary marginal distributions and captures dependence across variables via a latent multivariate Gaussian dependence structure.  Further, it can be applied as a joint model for mixed-type data, i.e., continuous, ordinal, and categorical  variables.

Often, analysts work with data that are subject to requirements to protect the confidentiality of data subjects and their sensitive attributes.  The literature on data privacy has shown that releasing results of any  statistical analysis leaks information about the underlying data values \citep{dwork2014algorithmic}.  Given enough results from the confidential data, analysts may be able to use released outputs to learn sensitive information \citep{dinur2003revealing, dwork2017exposed, abowd2022topdown}.  Hence, data analysts may seek, or even be required by data stewards, to limit the amount of information leakage.

One way to do so is to ensure released outputs satisfy differential privacy (DP)  \citep{dwork2006differential, calibratingnoise},  which has emerged as a gold-standard definition for privacy protection.  Indeed, there have been several methods for estimating differentially private Gaussian copulas.  \cite{li2014differentially} achieve DP by adding Laplace noise to pairwise correlations.  They present two algorithms, one based on a 
pseudo-likelihood estimator and another based on Kendall's $\tau$. 
\cite{asghar2020differentially} extend this framework using binary coding of categorical variables and discretizing continuous variables. They add
Laplace noise to the resulting indicator counts to reconstruct the copula correlations.
Finally, \cite{Wang2022locally} present a DP copula for a variant of DP called local differential privacy. 

These approaches come with a significant limitation: they do not provide reliable (differentially private) uncertainty quantification for the estimates of the copula correlations. Their algorithms release only noisy summary statistics, from which the correlation matrix is reconstructed. Plugging these privatized statistics into standard variance estimators, such as the sandwich estimator, generally leads to invalid inference, as it fails to account for both the sampling variability and the additional randomness introduced by the privacy mechanism.  Indeed, this limitation is the primary motivation for our work, namely to develop DP copula correlation methods that facilitate principled uncertainty quantification.

The extant methods also have particular features that may affect their usefulness in some applications.  The methods of \cite{li2014differentially} and 
\cite{Wang2022locally} could have 
high sensitivity in contexts where the privacy-mechanism must account for outliers. In turn, this could result in large variances in the Laplace distributions used in the DP mechanisms. The algorithm in \cite{li2014differentially} that uses Kendall's $\tau$ gets around this issue; however, Kendall's $\tau$ may not offer accurate estimates when the data contain ties. There are tie-adjusted variants of Kendall's $\tau$ to 
correct this bias, but using them increases the global sensitivity since each tied observation simultaneously can affect multiple concordant and discordant pairs. The binary-coding-based method of \cite{asghar2020differentially} may not perform well when some categories or discretized bins contain small counts.  In such cases, adding DP noise to the counts can produce negative values, which subsequently need to be clipped to zero.  This clipping can result in undesirable biases.  Their method also assumes access to publicly known quantiles or ranges of variables, which may not be available in some settings.  Finally, their method requires allocating part of the privacy budget to estimating marginals, even if only the dependence structure is of interest.

With these limitations in mind, we present an approach for estimating DP Gaussian copula correlations.  
To do so, we use a Bayesian modeling 
framework that explicitly incorporates the randomness of the DP noise mechanism into the inferential process. 
The resulting posterior distribution can be used 
to make credible intervals for the underlying true copula correlation parameters.
To the best of our knowledge, this is the first work to provide uncertainty quantification for DP copula correlations. In designing the DP algorithm, we 
coarsen the data into two-way tables of counts above or below medians, 
thereby reducing global sensitivity. 
We note that data coarsening \citep{heitjan1991ignorability, miller2019robust} has been used as a tool for efficient computation with Gaussian copulas  \citep{feldman2025using}, although not in the manner we use here.
We also develop a maximum likelihood estimator (MLE) of the copula correlation, which can be useful if one only seeks a point estimator. 
As part of this MLE algorithm, we design
two new DP noise mechanisms that ensure noisy counts fall within pre-specified bounds.
We examine repeated sampling properties of the proposed point estimators and posterior inferences using simulation studies.  
The simulations suggest that the proposed methods can offer lower mean-squared errors than 
existing approaches while also providing
intervals that propagate sources of uncertainty.

The remainder of this article is organized as follows. Section \ref{sec:preliminary} provides brief reviews of DP and Gaussian copulas.  Section \ref{sec:method} presents the Bayesian method for estimating the copula correlation under DP. Section \ref{sec:mechanism} introduces the MLE approach, along with one of the new range-preserving noise mechanisms for count queries; the other range-preserving mechanism is described in the supplementary material. Section \ref{sec:experiment} reports the results of the simulation studies, and Section \ref{sec:genuine} illustrates the methods using data about people's diets. We note that our focus in these empirical evaluations is on the accuracy of estimates of the copula correlations.  We do not consider synthetic data generation, although one could independently estimate the  marginal distributions in a DP manner, e.g., as done by \cite{asghar2020differentially}, to generate synthetic data. Finally, Section \ref{sec:discussion} concludes with a discussion of future research. Codes for all analyses are available at \url{https://github.com/shuowang7878/DPBayesCopula}.

\section{Background}
\label{sec:preliminary}

Section \ref{sec:gc} reviews the Gaussian copula, and Section \ref{sec:dp} reviews differential privacy. 

\subsection{Gaussian Copula}\label{sec:gc}

For $j=1, \dots, p$, let $X_j$ represent one of the study variables of interest.  Let $X = (X_1, \dots, X_p)$. Each $X_j$ has some true marginal cumulative distribution function $F_j$.   For now,  we presume the variables are modeled as continuous; 
see \citet{feldman2024nonparametric} for adaptations of the Gaussian copula model to incorporate categorical variables.
To capture multivariate dependence  in the Gaussian copula, we introduce a latent variable $Z_j$ for each $X_j$ with $\mathbb{E}(Z_j)=0$ and $\operatorname{Var}(Z_j)=1$.  Let  
$Z = (Z_1, \dots, Z_p)$.  We presume these latent variables jointly follow a multivariate normal distribution with correlation matrix $R=\{R_{jj'}\}_{1\leq j<j'\leq p}$, where $R_{jj'}$ is the correlation between $Z_j$ and $Z_{j'}$. 
Putting it all together, we have the Gaussian copula model, specified as 
\begin{equation}
\label{eq:Gaussian_copula}
    Z  \sim \mathcal{N}_p(\mathbf{0},R), 
    \qquad 
     X_j = F_j^{-1}(\Phi(Z_j)) \quad \text{for } j = 1, \dots, p,
\end{equation}
where 
$\Phi(\cdot)$ is the standard normal cumulative distribution function, and each 
 $F_j^{-1}(u)=\inf\{x:F_j(x)\geq u\}$ where $u \in [0,1]$.

This construction decouples the marginal distributions of $X$ from the dependence structure encoded by $R$, allowing the correlation structure to be analyzed independently of marginal features such as skewness, heavy tails, or bounded support.  Typically, each $F_j$ is estimated from the observed data. However, in this article, we forego estimation of the marginal distributions, as 
our DP algorithms are able to estimate $R$ without estimating  $(F_1, \dots, F_p)$.

The $R$ is the Pearson's correlation of $Z$ and hence  measures linear association among the variables after mapping their marginal distributions to standard normal distributions.
For example, a large positive value of $R_{jj'}$ indicates that $X_j$ and $X_{j'}$ tend to move together in terms of their relative positions within their individual marginal distributions; large negative values indicate movement in opposite directions. At the extremal value  $R_{jj'}=1$, $X_j$ is an almost surely non-decreasing function of $X_{j'}$, so that the quantiles of each variable align exactly.
When $R_{jj'}=0$, $Z_j$ and $Z_{j'}$ are independent, indicating no monotonic dependence between $X_j$ and $X_{j'}$. As these examples suggest, $R$ summarizes the monotone associations among the variables in ways that remain meaningful for non-normal data.

Much of the recent research on Gaussian copulas has focused on developing computationally convenient estimation algorithms. For example,  \cite{d2007extending} introduces the extended rank likelihood for Bayesian inference on the Gaussian copula correlation, and \cite{murray2013bayesian} extend this work to high-dimensional settings using factor models. The rank likelihood  enables fully Bayesian inference on the copula dependence structure without having to specify priors on the marginal distribution functions. As such, convenient Gibbs sampling algorithms can be developed for posterior inference. 
As more recent examples, \cite{feldman2022bayesian} introduce the rank-probit likelihood to extend the Gaussian copula estimation to accommodate unordered categorical variables, and \cite{feldman2024nonparametric} use the rank-probit likelihood for estimation of a Bayesian Gaussian mixture copula.

\subsection{Differential Privacy}\label{sec:dp}

Let $D$ represent some dataset comprising $n$ individuals measured on $p$ variables.  DP utilizes the concept of neighboring datasets, which we define as follows.  Let $D'$ be a dataset also with $n$ individuals and $p$ variables. Then, $D'$ is a neighboring dataset of $D$ if it differs from $D$ by the substitution of a single individual's data.  That is, $n-1$ of the individuals are the same in $D$ and $D'$, but one individual is different. This definition of neighboring datasets implies that the sample size $n$ of $D$ is considered public.  One also can define neighboring datasets via the insertion or deletion of one individual from $D$, although we do not do so here. Definition \ref{def:DP} provides the definition of DP that we use in our work.

\begin{definition}[$\epsilon$-differential privacy]
\label{def:DP}
Let $D$ and $D'$ be any neighboring datasets. A randomized algorithm $\cA$ satisfies $\epsilon$-differential privacy ($\epsilon$-DP) if, for every measurable set $S\subseteq \operatorname{Range}(\cA)$,
\begin{equation}\label{eq:DPdef}
    \Pr(\cA(D)\in S)\leq e^\epsilon\Pr(\cA(D')\in S).
\end{equation}\label{eq:DPdef}
\end{definition}
The probabilities in Definition~\ref{def:DP} 
are 
taken with respect to the randomness of the algorithm $\cA$ alone,  not over any sampling distribution for $D$.
 
 The criterion in Definition \ref{def:DP} provides a probabilistic guarantee that changing any single individual's information has a controlled effect on the algorithm's output. The degree of control is governed by the parameter $\epsilon>0$, known as the privacy budget. Smaller values of $\epsilon$ ensure that analysts, well-intentioned or not, cannot easily discern from the output of $\mathcal{A}$ whether $S$ was generated using $D$ or $D'$, thereby making it difficult to learn if any particular individual was in the data.  Larger values of $\epsilon$ offer less of a guarantee.  However, typically there is a trade-off in choosing $\epsilon$.  For most $\mathcal{A}$ that satisfy DP, decreasing $\epsilon$ results in greater distortion of the confidential data analysis.  Typical recommendations in the literature involve setting $\epsilon \leq 1$, although in practice (much) larger values are often used \citep{Kazan2024prior}.

 A common way to construct an $\epsilon$-DP algorithm is to 
 add random noise to the output of the analysis of the confidential data $D$. The scale of the noise is determined by how sensitive the computation is to changes in a single record in the worst case. This quantity is known as the $\ell_1$ sensitivity and is defined in Definition \ref{mechanism:sensitivity}.

\begin{definition}[$\ell_1$-sensitivity]
\label{mechanism:sensitivity}
The $\ell_1$ sensitivity of a function $M$ is 
    $\Delta(M)=\max_{D\sim D'}\norm{M(D)-M(D')}{1}$,
where $D\sim D'$ denotes neighboring datasets.
\end{definition}
For example, when $M(D)$ counts the number of successes out of $n$ trials, $\Delta(M)=1$ since changing one individual at most can increase or decrease the count $M(D)$ by one.  

An example of a DP mechanism for adding integer-valued noise to count data is 
the geometric mechanism \citep{ghosh2012universally}, which we define in Definition \ref{mechanism:geom}.

\begin{definition}[Geometric mechanism]
\label{mechanism:geom}
Let $M$ be a counting query that outputs $M(D)\in\mathbb{N}$ with $\ell_1$-sensitivity $\Delta$. The geometric mechanism outputs $\wt{M}_{\text{Geom}}(D)=M(D)+\delta$, where $\delta\in\mathbb{Z}$ is a draw from the double-geometric distribution,
\begin{equation}\label{eq:geom}
    \Pr(\delta = k) = 
    \left(\frac{1 - e^{-\epsilon / \Delta}}{1 + e^{-\epsilon / \Delta}}\right) 
    e^{-\epsilon|k|/\Delta}, 
    \quad \text{for } k \in \mathbb{Z}.
\end{equation}
\end{definition}

It is possible for $\wt{M}_{\text{Geom}}(D) <0$  or  $\wt{M}_{\text{Geom}}(D) > n$,  
which are incompatible with count queries for fixed $n$.  This can be problematic for estimation methods that depend on counts being in the feasible region. In  Section~\ref{sec:mechanism}, we discuss  post-processing techniques  that ensure the ultimately-used noisy counts are in the feasible region.  We note that the Bayesian inference methods of Section \ref{sec:method} do not require enforcement of constraints on $\wt{M}_{\text{Geom}}(D)$.

DP has two useful properties that we leverage in developing algorithms, stated here as Proposition \ref{prop:seq} and Proposition \ref{prop:postprocess}.  See \citet{dwork2014algorithmic} for proofs of these propositions.
\begin{proposition}[Sequential composition]\label{prop:seq}
   For $k=1,\ldots,m$, let $\mathcal{A}_k$ be an $\epsilon_k$-DP algorithm. Then the joint algorithm $(\mathcal{A}_1, \ldots, \mathcal{A}_m)$, which on input $D$ releases the tuple $(\mathcal{A}_1(D), \ldots, \mathcal{A}_m(D))$, is $\sum_{k=1}^m \epsilon_k$-DP.
\end{proposition}
\begin{proposition}[Post-processing property]
\label{prop:postprocess}
If $h$ is any randomized mapping independent of $D$, and $\cA$ is an $\epsilon$-DP algorithm, then their composition $h \circ \cA$ is also $\epsilon$-DP.
\end{proposition}

\section{Bayesian Estimation of Copula Correlations}
\label{sec:method}

In this section, we present the Bayesian de-noising approach to estimating the copula correlation. In Section \ref{sec:statistic}, we describe our strategy for turning the values in $D$ into a series of two-way tables of counts above and below the median of each variable.  In Section \ref{sec:estimation}, we describe how to make DP versions of those counts.  In Section \ref{sec:noiseaware}, we present a Bayesian post-processing procedure for estimation.  We begin the presentation with methods for estimating single $R_{jj'}$ and then discuss how to extend to multivariate $R$.  

For $i=1, \dots, n$ and $j=1, \dots, p$, let $x_{ij}$ be the value of $X_j$ for individual $i$.  Let $D = \{x_{ij}: 1 \le i \le n, 1 \le j \le p\}$ denote the observed data.
We presume $D$ follows the Gaussian copula model in \eqref{eq:Gaussian_copula}.

\subsection{Characterizing the Copula Correlation with Two-Way Tables}
\label{sec:statistic}
We first define a statistic $t_{jj'}$ that counts the number of observations that exceed  the sample medians of both $X_j$ and $X_{j'}$.  For any pair of variables $1 \leq j<j'\leq p$, we define 
\begin{equation}
\label{eq:statistic}
    t_{jj'} = 
    \sum_{i=1}^n \mathbb{I}\left(x_{ij}\geq\operatorname{med}(X_j), ~ x_{ij'}\geq \operatorname{med}(X_{j'})\right),
\end{equation}
where $\mathbb{I}(\cdot)$ denotes the indicator function and $\operatorname{med}(X_j)$ denotes the sample median of $\{x_{ij}: i = 1, \dots, n\}$.  
We first consider $D$ such that, for any $j=1, \dots, p$ and for $i=1, \dots, n$, at most one $x_{ij} = \operatorname{med}(X_j)$ at the median. We discuss how to handle ties in Section \ref{sec:estimation}.

We use these counts for the DP algorithms because they offer a substantial reduction in sensitivity; see Section \ref{sec:estimation}.  
Furthermore, 
$t_{jj'}$ is invariant to monotone transformations, thereby ensuring robustness to outliers and providing protection against attacks targeted at the medians. 
Finally, although $t_{jj'}$ is a coarse summary of the observed data, it provides sufficient information to infer each pairwise correlation coefficient $R_{jj'}$.

Let $T_{jj'}$ be the random variable corresponding to the process that generates $t_{jj'}$. Under the Gaussian copula, $T_{jj'}$ admits a known distribution that can be parameterized in terms of $R_{jj'}$. This distribution does not depend on  $F_{j}$ and $F_{j'}$, allowing us to  target inference for $R_{jj'}$ without estimating models for the marginal distributions.

We now characterize the distribution  of $T_{jj'}\mid R_{jj'}$.   For notational simplicity, we assume $n$ is even so that exactly 
$n/2$ observations lie above and below the median.  When $n$ is odd, we allow $(n+1)/2$ observations to lie above the median.  For $i=1, \dots, n$, the monotone marginal transformation between $X$ and $Z$ defined in \eqref{eq:Gaussian_copula} implies that 
\begin{eqnarray}
\sum_{i=1}^n \mathbb{I}\left(
x_{ij}\geq\operatorname{med}(X_j), ~ x_{ij'}\geq \operatorname{med}(X_{j'})
\right)
&=&
\sum_{i=1}^n \mathbb{I}\left(
z_{ij}\geq\operatorname{med}(Z_j), ~ z_{ij'}\geq \operatorname{med}(Z_{j'})
\right)\\
&\approx&
\sum_{i=1}^n \mathbb{I}\left(
z_{ij}\geq0, ~ z_{ij'}\geq0
\right).
\end{eqnarray}
The approximation holds because the sample median $\operatorname{med}(Z_j)$ converges to its population median 0 as $n\rightarrow\infty$. This relationship holds for arbitrary $F_{j}$ and $F_{j'}$.

The approximation induces a $2\times 2$ contingency table in the latent space by dichotomizing $Z_j$ and $Z_{j'}$ using indicators for above and below zero, with fixed row and column totals equal to $n/2$. The corresponding cell probabilities $p_{uv}$ where $u,v\in\{0,1\}$, can be expressed in terms of $R_{jj'}$ as
\begin{align}{}
p_{11}(R_{jj'})
    &= \Pr(Z_j \ge 0, Z_{j'} \ge 0)
     = \frac{1}{4} + \frac{1}{2\pi}\arcsin(R_{jj'})
     = p_{00}(R_{jj'})\\
p_{01}(R_{jj'})
    &= \Pr(Z_j < 0, Z_{j'} \ge 0)
     = \frac{1}{4} - \frac{1}{2\pi}\arcsin(R_{jj'})
     = p_{10}(R_{jj'}).
\end{align}

The expression for $p_{11}(R_{jj'})$ is the orthant probability of the standard bivariate normal distribution, originally derived by \citet{sheppard1899}. The derivation is provided in Section~S.5 of the supplementary material.

These cell probabilities uniquely identify each pairwise copula correlation coefficient, since the latent Gaussian assumption provides a one-to-one map between quadrant probabilities in $\mathbb{R}^{2}$, i.e., $p_{uv}$, and $R_{jj'}$. Thus, each $R_{jj'}$ is identified by its corresponding $t_{jj'}$, which is key for consistent estimation through likelihood-based procedures  \citep{feldman2025using}.

In this case, conditional on the fixed row and column totals,
$T_{jj'}$ follows a noncentral hypergeometric distribution with odds ratio defined by the cell probabilities, 
\begin{equation}
    \frac{p_{11}(R_{jj'})p_{00}(R_{jj'})}{p_{10}(R_{jj'})p_{01}(R_{jj'})}
    =\left(\frac{\pi+2\arcsin(R_{jj'})}{\pi-2\arcsin(R_{jj'})}\right)^2.
\end{equation}
Formally, for any possible realized count $t$, we have
\begin{equation}
\label{eq:PMF_T}
    \Pr(T_{jj'}=t\mid R_{jj'}) = 
    \frac{1}{C_{jj'}} 
    \binom{n/2}{t}^2  
    \left( \frac{\pi + 2 \arcsin(R_{jj'})}{\pi - 2 \arcsin(R_{jj'})} \right)^{2t},
\end{equation}
where the normalizing constant 
\begin{equation}
C_{jj'} = \sum_{k=0}^{n/2} \binom{n/2}{k}^2  \left( \frac{\pi + 2 \arcsin(R_{jj'})}{\pi - 2 \arcsin(R_{jj'})} \right)^{2k}.
\end{equation}
The distribution of $T_{jj'}$ thus belongs to one-parameter exponential family with natural parameter $\eta_{jj'}=2\log((\pi + 2 \arcsin(R_{jj'}))/(\pi - 2 \arcsin(R_{jj'})))$.

\subsection{Ensuring Differential Privacy}
\label{sec:estimation}

To ensure differential privacy, we first perturb $t_{jj'}$
via the geometric mechanism. 
We subsequently use the resulting noisy count to privately estimate $R_{jj'}$. The noise scale is determined by the  $\ell_1$-sensitivity  defined in Theorem \ref{thm:sensitivity}.
\begin{theorem}
\label{thm:sensitivity}
    Let $D=\{x_{ij}\}\in\mathbb{R}^{n\times p}$  have no ties at the medians; that is, for any $j=1, \dots, p$, the 
    $\sum_{i=1}^n \mathbb{I}\left(x_{ij}=\operatorname{med}(X_j)\right) \leq 1$. Consider any neighboring dataset $D'=\{x'_{ij}\}\in\mathbb{R}^{n\times p}$ that has no ties at the medians. Then, for any $1\leq j<j'\leq p$, the $\ell_1$-sensitivity of $t_{jj'}$ is $\Delta(t_{jj'})=1$.
\end{theorem}

\begin{proof}
    For $j=1,\dots,p$, define the unit-level indicators 
$a_{ij}=\mathbb{I}\left(x_{ij}\geq\operatorname{med}(X_j)\right)$ and $a'_{ij}=\mathbb{I}\left(x'_{ij}\geq\operatorname{med}(X'_j)\right).$
    Let $l$ be the index of the individual on which $D$ and $D'$ differ. Since there are no ties at the medians, changing the data for individual $l$ can affect the membership of at most two elements in each set $\{i:x_{ij}\geq\operatorname{med}(X_j)\}$. These include individual $l$ and one individual $u_j$ whose $x_{u_jj}$ is adjacent to $\operatorname{med}(X_j)$, that is, $u_j= a_{lj} \lceil \frac{n-1}{2} \rceil + \left(1-a_{lj}\right) \lceil\frac{n+1}{2} \rceil$. Thus, for any $j=1,\dots,p$, we have $a'_{ij}=a_{ij}$ for all 
        $i \notin \{l,u_j\}.$
    Since $\sum_{i=1}^n a_{ij}=n/2=\sum_{i=1}^n a'_{ij}$, it follows that
        $a'_{lj}-a_{lj}=-\left(a'_{u_jj}-a_{u_jj}\right)$, for any $j=1,\dots,p$.

    Consider now the difference in $t_{jj'}$ in the neighboring datasets $D$ and $D'$. We have
    \begin{align}
    t_{jj'}(D')-t_{jj'}(D)
    &=\sum_{i=1}^n a'_{ij}a'_{ij'}-a_{ij}a_{ij'}\\
    &=\sum_{i=1}^n\left[a'_{ij}\left(a'_{ij'}-a_{ij'}\right)
    +a_{ij'}\left(a'_{ij}-a_{ij}\right)\right]\\
    &= 
    \sum_{i\in\{l,u_{j'}\}}a'_{ij}\left(a'_{ij'}-a_{ij'}\right) + 
    \sum_{i\in\{l,u_j\}}a_{ij'}\left(a'_{ij}-a_{ij}\right)\\
    &= 
    \left(a'_{lj}-a'_{u_{j'}j}\right)\left(a'_{lj'}-a_{lj'}\right)
    +\left(a_{lj'}-a_{u_jj'}\right)\left(a'_{lj}-a_{lj}\right).
\end{align}
If either $a'_{lj'}-a_{lj'}=0$ or $a'_{lj}-a_{lj}=0$, then $|t_{jj'}(D')-t_{jj'}(D)|\leq|1|\cdot|1|=1$. Otherwise, when both pairs differ, enumerating the remaining four possible configurations of $a'_{lj},a_{lj},a'_{lj'},a_{lj'}$ yields
\begin{equation}
t_{jj'}(D')-t_{jj'}(D)\in\left\{\pm(1-a'_{u_{j'}j}-a_{u_jj'}),~\pm(a'_{u_{j'}j}-a_{u_jj'})\right\}.
\end{equation}
Hence, $|t_{jj'}(D')-t_{jj'}(D)|\leq 1$ and the bound is tight. Therefore, 
$\Delta(t_{jj'})=1$.
\end{proof}

Importantly, to construct $t_{jj'}$ we do not need to allocate privacy budget to estimate the medians. We simply 
fix each marginal total at $n/2$ (adjusted as needed when $n$ is odd), which incurs no privacy cost since $n$ is public.
Thus, we can dedicate the privacy budget to quantities needed to estimate the copula correlations.

For marginal distributions that can have ties with non-zero probability, one can enforce uniqueness by adopting a data-independent, lexicographic tie-breaking rule that does not compromise privacy. Specifically, prior to accessing any data, for each variable $X_j$ considered as potentially having ties, generate a $n\times 1$ vector of random keys, $k_{1j}, \dots, k_{nj}$, sampled independently from $k_{ij} \sim N(0,1)$ and set before the data are collected.
Then, define a lexicographic ordering on the pairs $(x_{ij}, k_{ij})$ as
$\left(x_{ij}, k_{ij}\right) \succeq \left(x_{i'j}, k_{i'j}\right)$ if either  
$x_{ij} > x_{i'j}$, or 
($x_{ij} = x_{i'j}$ and $k_{ij} \ge k_{i'j}$).
We define $\operatorname{med}(X_j)$ with respect to the ordering of the pairs. Under this construction, we can allocate each $x_{ij}$ as above or below the median, so the assumptions of Theorem~\ref{thm:sensitivity} hold and
$\Delta(t_{jj'})=1$.  

Because the random keys are fixed in advance and independent of the data, this strategy for managing ties does not affect the definition of neighboring databases and hence does not impact the privacy guarantee. 
Additionally, this strategy can be applied regardless of whether the data contain ties. If no ties exist, the induced ordering coincides with the original ordering of the records. As a result, inspection of $D$ to determine which $X_j$ have ties is not required, thereby avoiding any additional privacy loss from peeking at $D$.

Utilizing the lexicographic tie-breaking rule
could affect the estimates of the 
dependence structure.
However, as long as the number of ties at the median is not large, the effect 
should be relatively benign. 
In the supplementary material, we present simulations illustrating that  the tie-breaking strategy can perform well for data with moderate amounts of ties at the median, and that it can perform well 
for ordinal data with a small number of levels.

Finally, we note that when we apply the geometric mechanism independently to each of the $p(p-1)/2$ counts $t_{jj'}$, we can use Proposition \ref{prop:seq} to show that the overall 
$\ell_1$ sensitivity of the collection $\mathcal{T} = \{t_{jj'}\}_{j<j'}$ is bounded by  $\Delta\left(\mathcal{T}\right)=p(p-1)/2.$

\subsection{Bayesian Post-processing Inference for Copula Correlations} \label{sec:noiseaware}

We next present the Bayesian post-processing method for obtaining posterior intervals for any $R_{jj'}$ and ultimately for $R$. We call it the Bayesian noise-aware method, abbreviated as Bayes-NA.  
The basic strategy is to treat each $t_{jj'}$ and ultimately $\mathcal{T}$ as  nuisance parameters, and estimate $R_{jj'}$ and ultimately $R$ via a marginal likelihood that averages over possible values of $T_{jj'}$. 

For any pair of variables $(X_j, X_{j'})$, let $\widetilde{T}_{jj'}$ be the random variable corresponding to 
the geometric mechanism applied to $t_{jj'}$.  
The marginal likelihood under the geometric mechanism given the realized noisy count $\tilde{t}_{jj'}$ is
\begin{equation}
\begin{aligned}
    \Pr(\widetilde{T}_{jj'} = \tilde t_{jj'}\mid R_{jj'}) &= 
\sum_{t=0}^{n/2}\Pr(\widetilde{T}_{jj'}=\tilde t_{jj'}\mid T_{jj'}=t)\Pr(T_{jj'}=t\mid R_{jj'}) \\ 
    & = \sum_{t=0}^{n/2}
    \frac{1 - e^{-\frac{\epsilon}{\Delta}}}{1 + e^{-\frac{\epsilon}{\Delta}}} 
     e^{-\left| \tilde{t}_{jj'} - t \right|  \frac{\epsilon}{\Delta}}
    \frac{\binom{n/2}{t}^2 \left( \frac{\pi + 2 \arcsin(R_{jj'})}{\pi - 2 \arcsin(R_{jj'})} \right)^{2t}}
    {\sum_{k=0}^{n/2} \binom{n/2}{k}^2 \left( \frac{\pi + 2 \arcsin(R_{jj'})}{\pi - 2 \arcsin(R_{jj'})} \right)^{2k}}.
\end{aligned}\label{eq:likelihoodtj}
\end{equation}
When $p=2$ so that interest focuses solely on $R_{12}$, one can impose a prior distribution on $R_{12}$ such as the uniform distribution.  Simple methods like a grid sampler can be used for posterior inference.

Of course, Gaussian copulas typically are used with $p>2$ variables. Thus, we now extend to inference for $R$. We require the joint likelihood,  $\Pr(\widetilde{T}_{12}, \dots, \widetilde{T}_{p-1,p} \mid R)$. However, the collection of noisy statistics, $\widetilde{\mathcal{T}} = (\widetilde{T}_{12}, \dots, \widetilde{T}_{p-1,p})$, is a perturbed projection of an underlying $2^p$ contingency table. The joint likelihood is thus a sum over all possible $2^p$ cell configurations. This sum leads to an analytically intractable expression even for moderate $p$. 

To circumvent this computational challenge, we instead  use the composite likelihood, obtained by multiplying the contributions from \eqref{eq:likelihoodtj} for each $(X_j, X_{j'})$.  We have   
\begin{equation}
    f_{\mathrm{CL}}(\widetilde{\mathcal{T}}\mid R) = 
    \prod_{1\leq j<j'\leq p}
    \Pr(\widetilde{T}_{jj'}=\tilde t_{jj'}\mid R) = 
    \prod_{1\leq j<j'\leq p}
    \Pr(\widetilde{T}_{jj'}=\tilde t_{jj'}\mid R_{jj'}).
\end{equation}
This strategy neatly extends the method for  inference about a single correlation coefficient to the multivariate setting; analysts compile the $p(p-1)/2$ pairwise statistics $t_{jj'}$ and use Proposition \ref{prop:seq} for 
privacy accounting.
The use of composite likelihoods for DP post-processing is suggested by \citet{pistner}, who use it to estimate Bayesian latent class models for categorical data. Outside of privacy contexts, research has shown that composite likelihoods provide reasonable approximations for joint likelihoods like the Gaussian copula 
\citep{bai2014efficient, varin2011overview, ribatet2012bayesian}. We empirically assess the accuracy of the composite likelihood approximation in the supplementary material.

As a prior distribution for $R$, we use the LKJ prior \citep{lkjprior}, $\pi(R) = \mathrm{LKJ}(R \mid 1)$. The approximate posterior inference is  $p(R\mid \widetilde{\mathcal{T}}) \propto \pi(R)f_{\mathrm{CL}}(\widetilde{\mathcal{T}}\mid R)$.
Posterior sampling is straightforward to implement in the software package \textsf{Stan} using a No-U-turn sampler \citep{nouturn}, an adaptive variant of Hamiltonian Monte Carlo. The posterior expectation can be used as a point estimate for $R$. The support of the LKJ prior is the space of valid correlation matrices, that is, symmetric matrices with unit diagonal, off-diagonal entries in $[-1,1]$, and positive semi-definiteness. \textsf{Stan} performs Hamiltonian Monte Carlo directly on this manifold via its built-in Cholesky factor parameterization, so every posterior draw of $R$ satisfies all of these constraints. 
The set of valid correlation matrices is convex, so posterior summaries such as the posterior mean are themselves valid correlation matrices. No post-hoc projection or reparameterization is required.

\section{MLE of Copula Correlation}
\label{sec:mechanism}

Although the Bayesian post-processing procedure offers uncertainty estimates,
 some analysts may prefer to eschew the computations in favor of an MLE of $R$. In Section \ref{sec:mle}, we present such an estimate, once again starting with a method for a single $R_{jj'}.$ 

This MLE 
presumes the noisy two-way tables are internally coherent, that is, all the noisy statistics are non-negative and sum to the known marginal totals. This can be ensured using a truncated geometric mechanism \citep{ghosh2012universally}: add  unbounded geometric noise as in   Definition~\ref{mechanism:geom}, and then apply a post-processing step that remaps negative values to zero and values exceeding any imposed upper limit to that limit.  However, as we illustrate in the supplementary material, this can generate large spikes of probability mass at zero and at the upper limit, which in turn can affect the estimates of the underlying true count.  We therefore develop two other variants of the geometric mechanism for ensuring internally coherent tables with fixed marginal totals. We present one of these methods in Section \ref{sec:trunDP}.  The other is described in the supplementary material, along with guidance for selecting among these three mechanisms. 
We note that these DP algorithms could be used in other contexts where one requires bounded noisy statistics.

\subsection{Noisy MLE Algorithm}\label{sec:mle}

To determine the MLE, we work with the exponential family representation of the distribution of $T_{jj'}$ described in Section \ref{sec:statistic}. Its parameter $\eta_{jj'}$ satisfies the estimating equation, $t_{jj'} = \mathbb{E}_{\eta_{jj'}}(T_{jj'})$. As DP requires us to use $\widetilde{t}_{jj'}$ rather than $t_{jj'}$,  we plug $\widetilde{t}_{jj'}$ into the estimating equation.  Using the probability mass function from  \eqref{eq:PMF_T}, we set 
\begin{equation}
\label{eq:MLE}
   \widetilde{t}_{jj'} =  \left(\sum_{t=0}^{n/2} \binom{n/2}{t}^2  \left( \frac{\pi + 2 \arcsin(\hat{R}_{jj'})}{\pi - 2 \arcsin(\hat{R}_{jj'})} \right)^{2t}\right)^{-1} \left(\sum_{t=0}^{n/2} t  \binom{n/2}{t}^2  \left( \frac{\pi + 2 \arcsin(\hat{R}_{jj'})}{\pi - 2 \arcsin(\hat{R}_{jj'})} \right)^{2t}\right).
\end{equation}

Since the mapping $R_{jj'}\mapsto\eta_{jj'}$ is strictly increasing, \eqref{eq:MLE} admits a unique solution for $\hat{R}_{jj'}$ when $\widetilde{t}_{jj'}\in (0, n/2)$. The right-hand side of~\eqref{eq:MLE} is continuous and monotone in $\hat R_{jj'}$, so the equation can be efficiently solved via bisection or other root-finding algorithms. In the boundary cases, $\hat R_{jj'}=-1$ when $\widetilde{t}_{jj'}=0$, and $\hat R_{jj'}=1$ when $\widetilde{t}_{jj'}=n/2$.  No solution exists when $\widetilde{t}_{jj'}\notin [0, n/2]$. Hence, to guarantee the existence of an MLE, we require a range-preserving DP mechanism so that 
$\widetilde{t}_{jj'}\in[0, n/2]$, e.g., as in Section \ref{sec:trunDP}.
We refer to $\hat R_{jj'}$ obtained from \eqref{eq:MLE} as the noise-naive MLE of the copula correlation, which we abbreviate as MLE-NN. The MLE-NN is $\epsilon$-DP, as it uses only the differentially private counts without ever accessing the underlying confidential data.

Deriving an MLE for $R$ is more complicated.  
Since each $\hat R_{jj'}$ is estimated independently, the correlation matrix formed from $\{\hat R_{jj'}\}_{1 \leq j < j' \leq p}$
may not be positive semi-definite (PSD). To resolve this issue, we project the matrix form of $\{\hat R_{jj'}\}_{1 \leq j < j' \leq p}$ onto the nearest valid correlation matrix using Higham's algorithm~\citep{higham2002computing}, resulting in the final estimator which we label $\widehat R$. Higham's algorithm returns the unique matrix that is closest to the input in Frobenius norm, subject to being symmetric, having unit diagonal, and being positive semi-definite. Together these properties also imply off-diagonal entries in $[-1,1]$. Hence $\widehat R$ is guaranteed to be a valid correlation matrix. By Proposition \ref{prop:postprocess}, $\widehat R$ is DP, with Proposition \ref{prop:seq} guaranteeing a privacy budget  no larger than $(p(p-1)/2)\epsilon$. We note that the total privacy budget used for $\widehat R$ can be large, particularly when $p$ is large and $\epsilon$ is not small.  

\subsection{Truncated DP Mechanism}\label{sec:trunDP}

We now present a truncated version of the geometric mechanism that, in our simulations, offers more accurate inferences than the standard truncated geometric mechanism in \cite{ghosh2012universally}, as well as the other variant presented in the supplementary material.  As the mechanism can be used in contexts beyond estimating DP Gaussian copula correlations, we describe it using a  more general notation and arbitrary lower and upper bounds $[L,U]$. 

Let $M(D)$ be some arbitrary true count, which we abbreviate simply by  writing $M$. 
Suppose we apply the geometric mechanism from Definition \ref{mechanism:geom}, resulting in the noisy count $\wt M_{\mathrm{Geom}}(D)$, which we write simply as $\wt M_{\mathrm{Geom}}$.
Assuming a uniform prior distribution $\pi(M)\propto
\mathbb{I}(L\leq M\leq U)$, we can compute the posterior distribution, 
\begin{equation}
    \pi(M\mid \wt M_{\mathrm{Geom}})
    \propto e^{-\frac{\epsilon}{\Delta} \cdot |M-\wt M_{\mathrm{Geom}}|}  \mathbb{I}(L\leq M\leq U).
\end{equation}
 Thus, the maximum a posteriori (MAP) estimator of $M$ is 
$\widetilde{M}_{\mathrm{TGM}} = \min\left(U, \max(L, \widetilde{M}_{\mathrm{Geom}})\right)$,
which is exactly the output of the truncated geometric mechanism of \cite{ghosh2012universally}.  

It is well known that the MAP estimator minimizes Bayes risk under the 0-1 loss function, specifically,
\begin{equation}
    \widetilde{M}_{\mathrm{TGM}} = \arg\min_{\hat{M}} \mathbb{E}_{M \sim \pi(M \mid \widetilde{M}_{\mathrm{Geom}})} 
    \left[ \mathbb{I}(\hat{M} \ne M) \right].
\end{equation}
If instead we seek to minimize the mean squared error of the point estimator, we should use the posterior expectation as the Bayes estimator, i.e., 
\begin{equation*}
    \mathbb{E}(M\mid \wt M_{\mathrm{Geom}}) = \arg\min_{\hat{M}} \mathbb{E}_{M \sim \pi(M \mid \widetilde{M}_{\mathrm{Geom}})} 
    \left[ (\hat{M} - M)^2 \right].
\end{equation*}

This realization forms the basis of our alternative DP mechanism, which we call the Bayesian truncated geometric mechanism, abbreviated as BTGM.  The BTGM  outputs 
$\mathbb{E}(M\mid \wt M_{\mathrm{Geom}})$, denoted by $\widetilde{M}_{\mathrm{BTGM}}$. 
Formally, let $m=\widetilde{M}_{\mathrm{Geom}}$ and $\alpha=e^{-\epsilon/\Delta}$, where $\Delta$ is the sensitivity. Under the uniform prior distribution for $M$, the posterior expectation is
\begin{equation}
\label{eq:postexp}
\widetilde{M}_{\mathrm{BTGM}}(m)
=
\begin{cases}
L+\frac{\alpha}{1-\alpha}
\frac{1-(U-L+1)\alpha^{U-L}+(U-L)\alpha^{U-L+1}}
{1-\alpha^{U-L+1}}
& \text{if } m < L, \\
L+
\frac{(m-L)(1-\alpha^2)+\alpha^{m-L+1}-(U-L+1)\alpha^{U+1-m}+(U-L)\alpha^{U+2-m}}
{(1-\alpha)
\bigl(1+\alpha-\alpha^{m-L+1}-\alpha^{U+1-m}\bigr)}
& \text{if } L \leq m \leq U, \\
L+\frac{(U-L)-(U-L+1)\alpha+\alpha^{U-L+1}}
{(1-\alpha)\bigl(1-\alpha^{U-L+1}\bigr)}
& \text{if } m > U.
\end{cases}
\end{equation}

 We omit the derivation for brevity.  It involves sums of geometric series and of arithmetic-geometric series.  For MLE-NN, we set $L=0$ and $U=n/2$ or $U=(n+1)/2$, as needed.

BTGM is $\epsilon$-DP, as it applies the data-independent transformation \eqref{eq:postexp} to the output of the $\epsilon$-DP geometric mechanism and hence preserves the privacy guarantee by Proposition \ref{prop:postprocess}.
Note that $\widetilde{M}_{\mathrm{BTGM}}$ might not be an integer. If one seeks to release it directly, it can be rounded to the nearest integer. This is not necessary for MLE-NN, as we can plug non-integer values into the estimating equation.

\section{Simulation Studies}
\label{sec:experiment}

In this section, we present results of simulation studies of the DP Gaussian copula correlation estimation methods from Section \ref{sec:method} and Section \ref{sec:mechanism}.  
In all simulation scenarios,  we  simulate independent observations from Gaussian copulas.  In each run, we randomly generate the correlation matrix $R$ from a scaled Wishart distribution with $p+1$ degrees of freedom, where we vary the number of variables $p \in \{2, 5, 10\}$. Specifically, we first sample $W \sim \mathrm{Wishart}(p+1, I_p)$, where $I_p$ is the $p \times p$ identity matrix,
and then convert $W$ to a correlation matrix via the standard transformation $R_{jj'} = W_{jj'} / \sqrt{W_{jj}\, W_{j'j'}}$. 
The resulting $R$ is uniformly distributed over the space of $p \times p$ valid correlation matrices, equivalent to the LKJ(1) distribution \citep{lkjprior}. For $p=2$, this reduces to $R_{12} \sim \mathrm{Unif}(-1,1)$.  We consider results for $p=2$ in Section \ref{sec:utility} and for $p>2$ in Section \ref{sec:multivariate}.  We describe the sample sizes and marginal distributions $\{F_j: j = 1, \dots, p\}$ in those sections.

We compare our proposed methods with two approaches inspired by \cite{li2014differentially} and \cite{asghar2020differentially}.
To make more direct comparisons, we adopt the DPCopula-Kendall version in \cite{li2014differentially}, denoted as Li-Kendall, in which all privacy budgets are used to add noise to Kendall's $\tau$ coefficients without estimating the marginals. 
We also consider an  implementation of the method in \cite{asghar2020differentially}, which we call Asghar-dpc. Instead of discretizing continuous variables into multiple bins, we dichotomize them at the theoretical median of their marginal distributions (rounded to two decimal places). We allocate the entire privacy budget to the two-way contingency tables in Algorithm 2 of \cite{asghar2020differentially}, and we use the noisy row and column sums of each joint table as substitutes for the one-way marginals in their Algorithm 1. We note that adopting the method of \cite{asghar2020differentially} exactly as they present it results in less accurate estimation of the copula correlations, since some privacy budget is allocated elsewhere. 

We measure the accuracy of any point estimator of the correlation matrix using the element-wise mean absolute error (MAE). Let $\hat{R}^{(h)}_{jj'}$ be a point estimate of $R_{jj'}^{(h)}$ computed in simulation run $h$, where $h=1, \dots, 1000$. We have 
\begin{equation}
    \text{MAE}(\hat{R}) 
    = \sum_{h=1}^{1000} \frac{2}{p(p-1)} \sum_{1 \leq j < j' \leq p} \big| \hat R_{jj'}^{(h)} - R_{jj'}^{(h)} \big|/1000.
\end{equation}

In addition, we evaluate the interval estimates for Bayes-NA using the empirical coverage rates and average interval lengths.  For any simulation run $h=1, \dots, 1000$, let $\tilde Q^{(h)}_{jj'}$ represent the 95\% credible interval for $R^{(h)}_{jj'}$  computed using the posterior distribution in Section \ref{sec:method}.  The empirical coverage rate is 
\begin{equation}
    \text{Coverage}= \sum_{h=1}^{1000}
    \frac{2}{p(p-1)} \sum_{1 \leq j < j' \leq p} 
    \mathbb{I}\left(R_{jj'}^{(h)} \in \tilde Q^{(h)}_{jj'}\right)/1000. 
\end{equation}
Writing $|\tilde Q^{(h)}_{jj'}|$ for the difference between the lower and upper limits of the interval $\tilde Q^{(h)}_{jj'}$, the average interval length is 
\begin{equation}
\text{Length}= \sum_{h=1}^{1000} 
    \frac{2}{p(p-1)} \sum_{1 \leq j < j' \leq p} 
     |\tilde Q^{(h)}_{jj'}|/1000.
\end{equation}

\subsection{MAEs for Bivariate Case}\label{sec:utility}

We first analyze simulations with $p=2$. We set $F_1= 
\mathrm{Gamma}(2, 1)$ and $F_2 = 
\mathcal{N}(0, 1)$.  We consider $n\in\{50, 100, 200, 500, 1000\}$ and $\epsilon\in\{0.01, 0.1, 1\}$. The inclusion of the small $\epsilon=0.01$ also serves as a reference for the simulations with $p>2$. For example, when $p=5$  and $\epsilon=0.1$, the privacy budget per pairwise correlation is approximately $\epsilon / \binom{5}{2}=0.01$. For Bayes-NA, we generate $1000$ posterior samples after a burn-in of 1000 iterations.

\begin{figure}[t]

\centering{

\includegraphics[width=\linewidth]{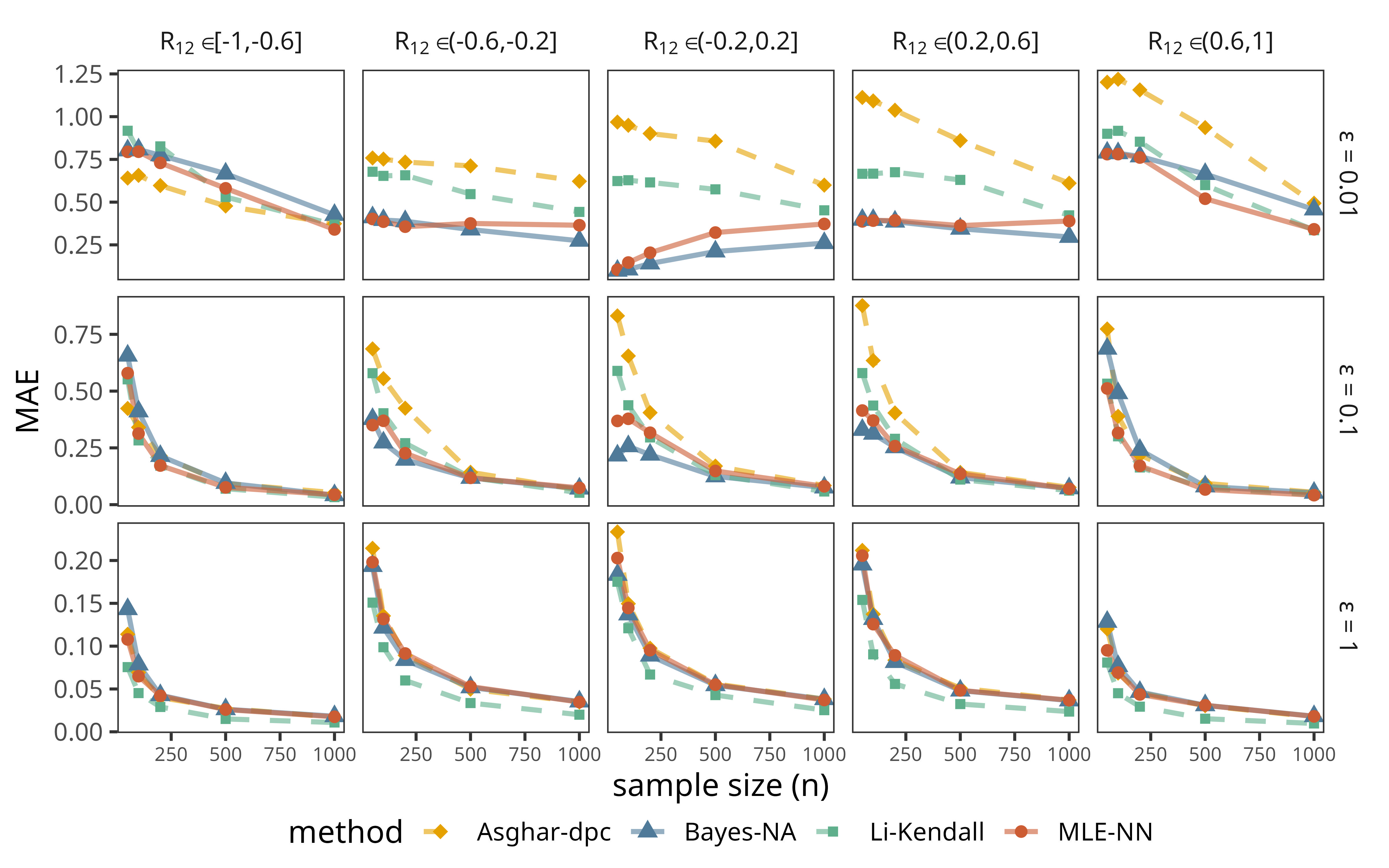}

}

\caption{\label{fig:MAE_p2_byrho}MAEs of $\hat R_{12}$  when $p=2$ and $n \in \{50, 100, 200, 500, 1000\}$ for $\epsilon \in \{0.01, 0.1, 1\}$.  Results based on 1000 simulation runs with a maximum Monte Carlo error of 0.05. Runs are grouped according to the simulated values of $R_{12}^{(h)}$. Solid lines: Bayes-NA, MLE-NN. Dashed lines: Li-Kendall, Asghar-dpc.  Note that the scale of $y$-axis differs across panels.}

\end{figure}

Figure \ref{fig:MAE_p2_byrho} displays the MAEs organized by binning the 1000 values of $R^{(h)}_{12}$ 
in consecutive bins of width 0.4.  Approximately 200 runs comprise each correlation bin.
Across most bins, Bayes-NA and MLE-NN attain the lowest MAEs, generally tracking one another  closely. Their MAEs tend to be largest for strong correlations, reflecting the difficulty of recovering extreme dependence under privacy constraints. Their MAEs are particularly small for correlations near zero when $\epsilon=0.01$ and $n \leq 200$,
a pattern not as evident for larger $\epsilon$. 
With small $n$ and $\epsilon$, the noise mechanism flattens the likelihood function for $R$. Thus, the posterior distribution is dominated by the prior distribution, and the posterior expectation concentrates around zero regardless of the true correlation. This also explains why the MAE near zero in Figure~\ref{fig:MAE_p2_byrho} does not decrease much with $n$. When the true correlation is itself near zero, this prior-dominated estimate is already accurate. For correlations farther from zero, the MAE does decrease with $n$ once the signal overcomes the DP noise.

In these highly challenging situations, analysts may want to use more informative prior distributions with Bayes-NA. For Asghar-dpc, the MAE increases steadily as correlations move towards the positive range. Because its noise mechanism remaps negative counts to 0, 
the resulting distribution of estimated correlations has a spike at $-1$, regardless of the true underlying correlation. This effect also explains the relatively lower MAE for correlations near $-1$ when $\epsilon = 0.01$.  For Li-Kendall, the MAE can be smaller than the MAEs for Bayes-NA or MLE-NN when $\epsilon=1$, in which case the DP algorithms do not add much noise. This is because Li-Kendall does not use coarsened data values.
However, these differences are typically around 0.01, which we expect to be negligible in practical contexts.

\subsection{MAEs for Multivariate Case}\label{sec:multivariate}

For the multivariate case, we consider $n \in \{200, 500, 1000\}$ and $\epsilon\in\{0.1, 0.5, 1, 5\}$. Here, $\epsilon$ is the total privacy budget aggregated over all $\binom{p}{2}$  noisy counts. When $p = 5$, we use marginal distributions
  $F_1 = \mathcal{N}(0, 1)$, $F_2=\mathrm{Exp}(1)$, $F_3 = \mathrm{Gamma}(2, 1)$,  $F_4 = \mathrm{Beta}(2, 5)$, and $F_5 = t_5$.
When $p=10$, we use 
   $F_1 = \mathcal{N}(0, 1)$, $F_2 = \mathcal{N}(1, 2)$, $F_3 = t_3$, $F_4 = t_{10}$, $F_5=\mathrm{Gamma}(1, 2)$, $F_6 = \mathrm{Gamma}(5, 2)$, $F_7 = \mathrm{Beta}(2, 5)$, $F_8 = \mathrm{Beta}(5, 2)$, $F_9 = \mathrm{Exp}(1)$, and  $F_{10} = \mathrm{Exp}(2)$.

\begin{figure}[t]

\centering{

\includegraphics[width=\linewidth]{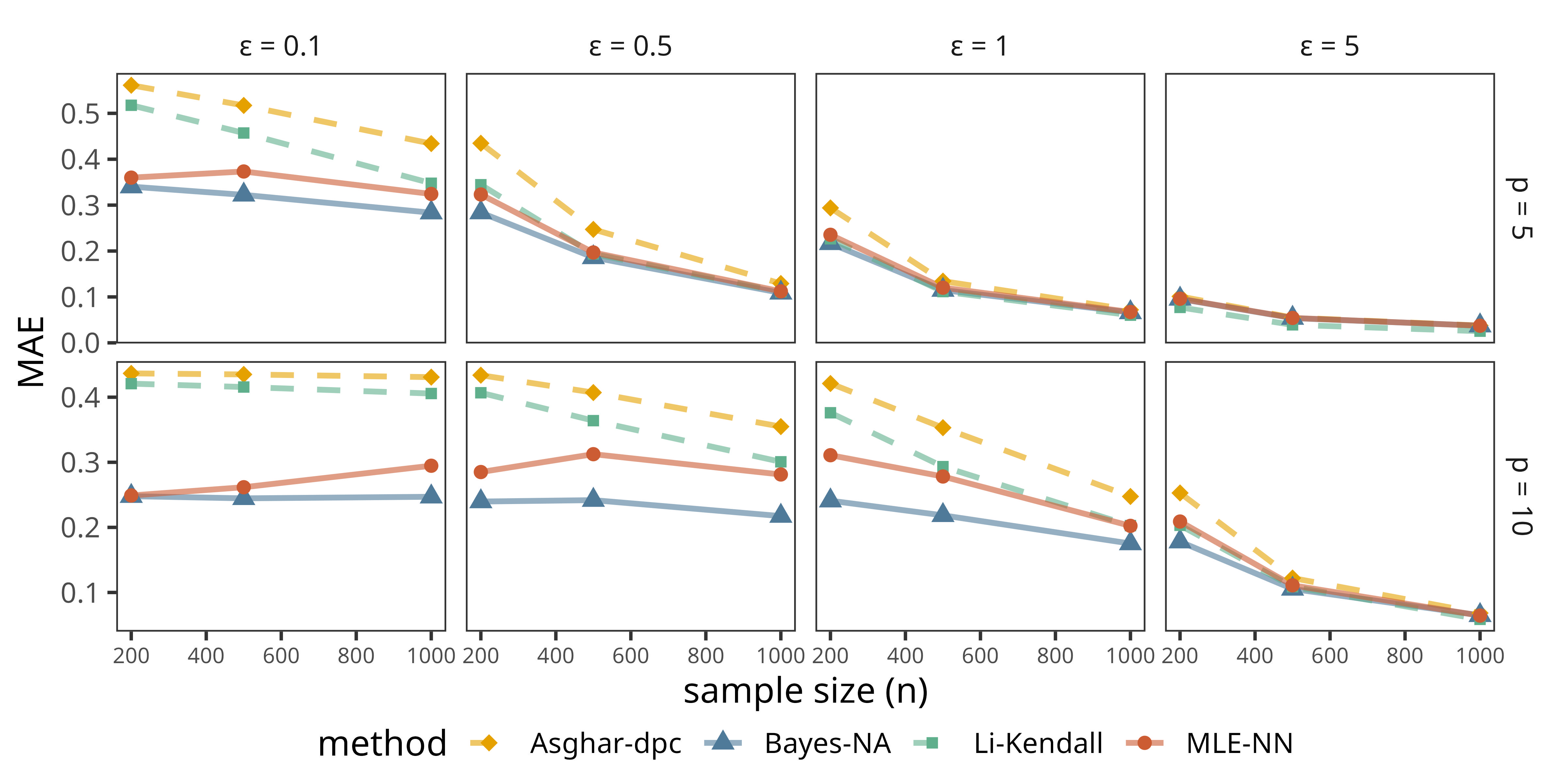}

}

\caption{\label{fig:MAE}MAEs of $\hat{R}$ when $p\in \{5, 10\}$ and $n \in \{200, 500, 1000\}$ for total privacy budget $\epsilon \in \{0.1, 0.5, 1, 5\}$. Results based on 1000 simulation runs with a maximum Monte Carlo error of 0.006. Solid lines: Bayes-NA, MLE-NN. Dashed lines: Li-Kendall, Asghar-dpc.}

\end{figure}

Figure \ref{fig:MAE} displays the MAEs over 1000 simulation runs of each setting.  Bayes-NA and MLE-NN tend to have MAEs that are smaller than or practically equal to  those for Asghar-dpc and Li-Kendall.
While the MAEs increase for all methods at these larger values of $p$ compared to $p=2$, MLE-NN and especially Bayes-NA maintain relatively stable performance, suggesting better scalability in higher-dimensional settings.  
As expected, increasing $n$ leads to improved accuracy, with MAEs gradually converging towards zero. When the per-correlation privacy budget $\epsilon / \binom{p}{2}$ is small, a larger sample size is required before the reduction in error becomes evident. In fact, over a range of small sample sizes, the error curve may remain flat or even slightly increase as sample size grows. This pattern is consistent with what we observe in the bivariate case when the true correlation is small, where estimates with small $n$ and $\epsilon$ tend to be pulled towards zero.

\subsection{Coverage Analysis for Bayes-NA}\label{sec:freq}

Bayes-NA and MLE-NN offer similar results with MAEs that are better or practically the same as the MAEs for Asghar-dpc and Li-Kendall.  However, a major advantage
 of Bayes-NA is that it produces posterior intervals for $R$. In this section, we use the simulations from Section \ref{sec:utility} and  Section \ref{sec:multivariate} to examine properties of these interval estimates.

 Table \ref{tab:coverage} displays the average coverage rates and interval lengths for the 95\% posterior intervals. 
 In most settings, the empirical coverage is near the nominal 95\% level. It rarely falls below 93\%. Notably, when $\epsilon=0.01$, the variance induced by the privacy mechanism can be so large that the intervals 
 are too wide to be of 
 practical use for correlation estimation. As $\epsilon$ and $n$ increase, 
 the average interval length decreases substantially. 

To assess the effect of dimensionality under comparable privacy conditions, we can compare the settings $(p=2,\epsilon=0.1)$, $(p=5, \epsilon=1)$ and $(p=10, \epsilon=5)$, all of which correspond to a 

\newpage
\begin{table}[H]
\begin{longtable}{@{}rrrrrrrrrrr@{}}
\caption{Empirical coverage rate (in \%) and average length of element-wise 95\% intervals for pairwise correlations produced by Bayes-NA. Results are averaged over 1000 runs.}%
\label{tab:coverage}\tabularnewline
\toprule\noalign{}
\multicolumn{1}{c}{ } & \multicolumn{5}{c}{Coverage} & \multicolumn{5}{c}{CI length} \\
\cmidrule(l{3pt}r{3pt}){2-6} \cmidrule(l{3pt}r{3pt}){7-11}
$\epsilon$ 
& 0.01 & 0.1 & 0.5 & 1 & 5
& 0.01 & 0.1 & 0.5 & 1 & 5 \\

\midrule\noalign{}
\endfirsthead

\toprule\noalign{}
\multicolumn{1}{c}{ } & \multicolumn{5}{c}{Coverage} & \multicolumn{5}{c}{Length} \\
\cmidrule(l{3pt}r{3pt}){2-6} \cmidrule(l{3pt}r{3pt}){7-11}
$\epsilon$ 
& 0.01 & 0.1 & 0.5 & 1 & 5
& 0.01 & 0.1 & 0.5 & 1 & 5 \\

\midrule\noalign{}
\endhead

\bottomrule\noalign{}
\endlastfoot

\addlinespace[0.3em]
\multicolumn{11}{l}{$p=2$}\\
\hspace{1em}50 & 95.3 & 94.0 & 93.0 & 92.4 & 93.5 & 1.892 & 1.813 & 1.100 & 0.776 & 0.613\\
\hspace{1em}100 & 94.3 & 94.9 & 93.6 & 92.3 & 92.5 & 1.889 & 1.638 & 0.681 & 0.504 & 0.438\\
\hspace{1em}200 & 93.7 & 95.2 & 94.6 & 94.9 & 95.0 & 1.878 & 1.201 & 0.415 & 0.342 & 0.316\\
\hspace{1em}500 & 95.1 & 94.7 & 94.3 & 94.9 & 95.2 & 1.815 & 0.596 & 0.226 & 0.205 & 0.199\\
\hspace{1em}1000 & 94.6 & 93.2 & 95.6 & 95.5 & 96.1 & 1.632 & 0.317 & 0.153 & 0.144 & 0.142\\
\addlinespace[0.3em]
\multicolumn{11}{l}{$p=5$}\\
\hspace{1em}200 &  & 94.7 & 94.9 & 94.9 & 93.6 &  & 1.492 & 1.326 & 1.058 & 0.448\\
\hspace{1em}500 &  & 94.8 & 95.0 & 94.8 & 94.2 &  & 1.446 & 0.942 & 0.601 & 0.259\\
\hspace{1em}1000 &  & 94.9 & 95.3 & 94.4 & 93.9 &  & 1.323 & 0.590 & 0.346 & 0.177\\
\addlinespace[0.3em]
\multicolumn{11}{l}{$p=10$}\\
\hspace{1em}200 &  & 94.9 & 95.3 & 95.1 & 95.1 &  & 1.146 & 1.141 & 1.125 & 0.869\\
\hspace{1em}500 &  & 94.9 & 94.5 & 94.7 & 94.2 &  & 1.145 & 1.114 & 1.038 & 0.530\\
\hspace{1em}1000 &  & 94.9 & 95.1 & 95.2 & 93.8 &  & 1.141 & 1.039 & 0.860 & 0.321\\
\bottomrule
\end{longtable}
\end{table}

\noindent per-correlation privacy budget of about 0.1. The coverage rates and average lengths are nearly identical across these settings, suggesting that the method can scale  with dimension (provided the total privacy budget is acceptable).

In the simulations with $p>2$, the Bayes-NA intervals have lower-than-nominal coverage rates for 
$R_{jj'}\approx\pm 1$ when both $n$ and $\epsilon$ are small. Supporting results are provided in the supplementary material. In these cases, the prior distribution  dominates the likelihood, and the PSD constraints force the LKJ prior for $p>2$ to assign little probability mass to correlations near $\pm 1$. 
In contrast, when either $n$ or $\epsilon$ becomes sufficiently large, the likelihood dominates the prior distribution, and the credible intervals have near-nominal coverage rates even for extreme correlations.

\section{Illustration Using Dietary Data}\label{sec:genuine}

We illustrate the methods using 
dietary intake data from two publicly available sources. First, we use the National Health and Nutrition Examination Survey (NHANES) Dietary Component for the 2017--2018 cycle \citep{cdc2020nhanes}, from which we extract the Dietary Interview -- Total Nutrient Intakes data.  Second, we use  the 
Food Patterns Equivalents Database (FPED) provided by the U.S.\ Department of Agriculture's Food Surveys Research Group \citep{bowman2020food}. The two datasets are matched by respondent sequence number. We focus on the first-day dietary recall. We analyze $p=7$ variables including total energy intake in kilocalories (Kcal), ratio of family income to poverty (Income), total grains intake (Grain), total dairy intake (Dairy), total vegetable intake (Veg), total protein food intake (PF), and added sugar (Sugar). The detailed classification of each food item can be found on the FPED website. 
The dataset has $n=5820$ individuals after filtering out missing values. We estimate DP Gaussian copula correlations using $\epsilon\in\{1,3\}$ and compare the results of  Bayes-NA with the non-private extended rank likelihood estimator (ERL) of \cite{d2007extending}.

\begin{figure}[t]

\centering{

\includegraphics[width=\linewidth]{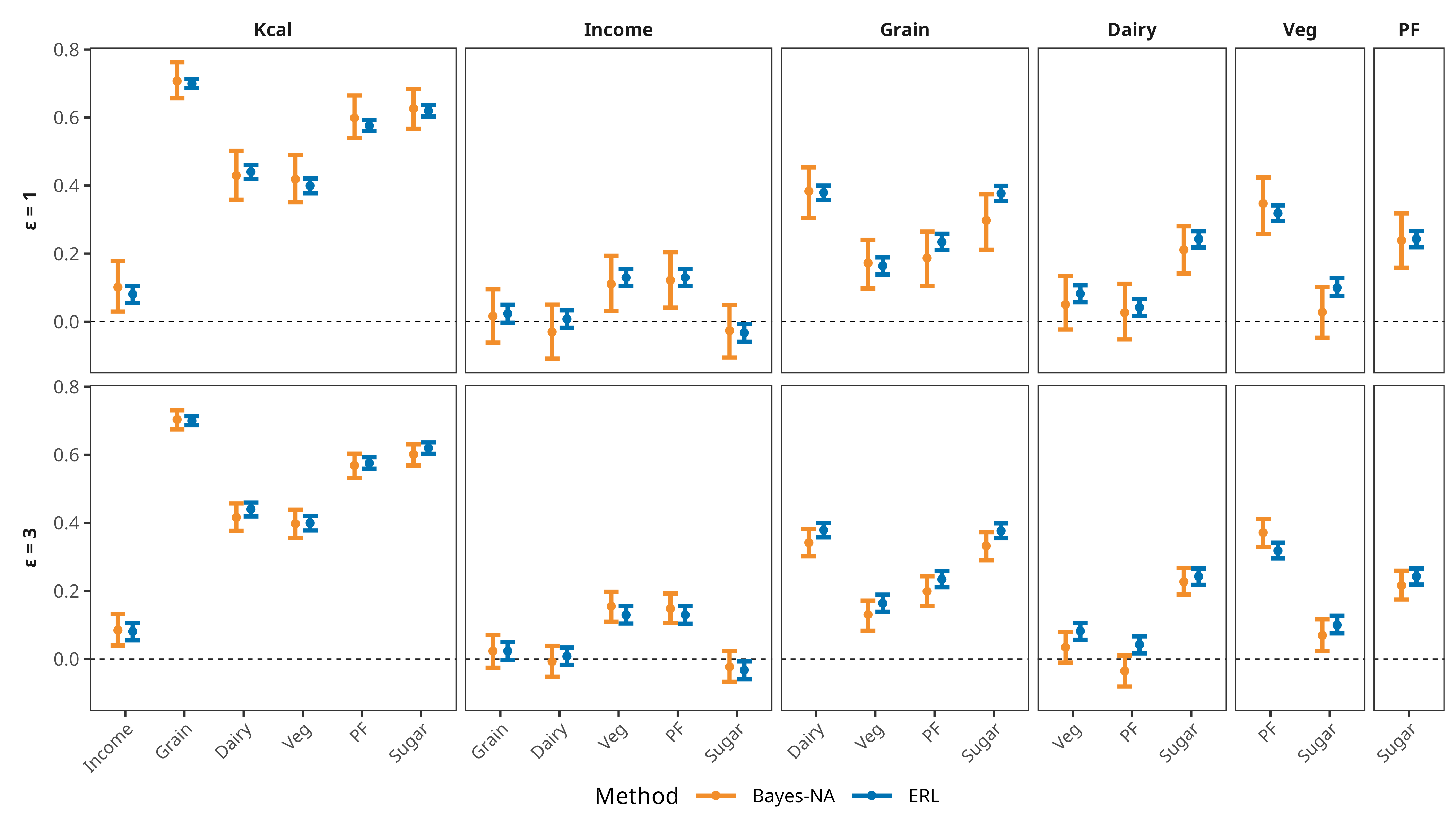}

}

\caption{\label{fig:CI_joint}Posterior summaries of pairwise correlations for NHANES and FPED data. Within each variable pair, the orange interval on the left is the Bayes-NA private estimate and the blue interval on the right is the non-private ERL estimate. Dots show posterior means and bars show the 2.5\% and 97.5\% posterior quantiles.}

\end{figure}

Figure~\ref{fig:CI_joint} summarizes the posterior distributions of the pairwise correlations among the seven variables based on 1000 posterior samples. 
The Bayes-NA interval estimates encapsulate the non-private ERL interval estimates. As $\epsilon$ increases, the credible intervals of Bayes-NA become narrower, as expected.  Overall, the resulting interpretations about associations are consistent for the private and non-private analyses.  For example, among all dietary components, both sets of results show that grain has the strongest positive correlation with total energy intake followed by added sugar, whereas the correlation with vegetables is relatively smaller. 
We note that posterior sampling via \textsf{Stan} is straightforward. In fact,  the runtime of Bayes-NA is faster than that of the ERL method. 

The posterior distribution of $R$ can be used for DP inferences for functionals of the correlations.  We illustrate these computations in the supplementary material, in which we  compute posterior distributions of regression coefficients using the posterior samples of $R$.

\section{Discussion}
\label{sec:discussion}

We propose a differentially private approach to characterizing dependence structures via the Gaussian copula model using low-sensitivity, coarsened statistics. We propose two estimation methods, a computationally efficient but noise-naive MLE, and a noise-aware Bayesian method that explicitly models the DP privacy mechanism. The latter provides 
interval estimates for the copula correlation and hence can serve as a more complete inferential framework for differentially private copulas. The simulation study shows that both methods can provide improved accuracy over extant approaches, especially for small $n$ and $\epsilon$. The Bayesian approach additionally provides reasonable coverage rates for interval estimation.

The estimator of $R$ does not consume any privacy budget on the marginals.
Combining our estimator with a separate DP procedure for marginal estimation thus offers a two-stage route to differentially private synthetic data.  There may be ways to use discretization to estimate the marginal counts with less privacy budget expenditure.  For example, if the analyst wishes to generate Gaussian distributed marginals, the analyst can  estimate two quantiles via  a DP mechanism and solve for the unique mean and variance of the Gaussian distribution that yields those quantiles.  For more general distributions, the analyst can estimate additional quantiles and fit splines between the points. 
We leave a careful study of this pipeline to future work.

Future work may also extend the proposed methodology to settings involving nominal categorical variables. One direction is to employ binary or multi-bit encodings, which would allow the coarsened statistics to be defined without requiring an inherent ordering and would remain compatible with the Bayesian inference framework developed here. Another direction is to generalize the Bayesian formulation to broader copula families to accommodate heavier-tailed or asymmetric dependence structures.

\section{Disclosure statement}\label{disclosure-statement}

The authors have
no conflicts of interest to declare.

\section{Data Availability Statement}\label{data-availability-statement}
The data that support the findings of this study are openly available at \url{http://www.ars.usda.gov/nea/bhnrc/fsrg} and \url{https://wwwn.cdc.gov/nchs/nhanes}.

\phantomsection\label{supplementary-material}
\bigskip

\begin{center}

{\large\bf SUPPLEMENTARY MATERIAL}

\end{center}

\begin{description}
\item[Supplementary material:] A newly proposed truncated geometric mechanism, algorithms for the proposed methods in the main text, additional simulation results, and additional analyses of the dietary data. (.pdf file)
\item[Source code:] simulation code and plotting functions. (zipped tar file)
\end{description}

\bibliography{bibliography.bib}

\end{document}


\def\spacingset#1{\renewcommand{\baselinestretch}%
{#1}\small\normalsize} \spacingset{1}


\if1\anon
{
  \title{\bf Supplementary Material for ``Differentially Private Bayesian Inference for Gaussian Copula Correlations''}
  \author{Shuo Wang\thanks{
    The authors gratefully acknowledge NSF-SES-2217456.}\hspace{.2cm}\\
    Department of Statistical Science, Duke University\\
    and \\
    Joseph Feldman \\
    Department of Statistics and Data Science, Washington University in Saint Louis\\ 
    and \\ 
    Jerome P. Reiter \\ Department of Statistical Science, Duke University}
  \maketitle
} \fi

\if0\anon
{
  \bigskip
  \bigskip
  \bigskip
  \begin{center}
    {\LARGE\bf Supplementary Material for Differentially Private Bayesian Inference for Gaussian Copula Correlations}
\end{center}
  \medskip
} \fi

This document contains supplementary material to the main text. Section \ref{sec:RGM} describes the additional variation of the truncated geometric mechanism referred to in Section 4.2 of the main text.  Section \ref{sec:algo} provides algorithmic summaries of Bayes-NA and MLE-NN.
Section \ref{sec:moresims} includes additional simulation results. Section \ref{sec:real_data} includes additional results for the analysis of the dietary data.   Section \ref{sec:orthant} contains a derivation of the cell probabilities that appear in (7) and (8) in the main text.

\section{Additional Truncated Geometric Mechanism}\label{sec:RGM}

Section~4.2 of the main text describes the DP algorithm BTGM, which ensures that a noisy count is bounded within analyst-specified bounds. We use BTGM in the MLE-NN algorithm. In this section, we present an alternative DP algorithm that ensures a noisy count falls within analyst-specified bounds. We describe the mechanism presuming that the bounds are generically $L$ and $U$.

The standard truncated geometric mechanism (abbreviated as TGM) and the BTGM mechanism described in the main text concentrate the probability mass of out-of-bounds values onto two points within $[L,U]$.  This leads to two probability spikes at the boundaries, as shown in Figure \ref{fig:mechanism}. An alternative algorithm is  to redistribute the out-of-range probability mass across the interval $[L,U]$ by normalization. We call this algorithm the renormalized geometric mechanism, which we abbreviate as RGM.

\begin{definition}[Renormalized geometric mechanism]
For any dataset $D$, let $M(D)$ be a true count with $\ell_1$-sensitivity $\Delta$ as defined in Definition 2.2 in the main text. Let $\epsilon'>0$ denote the solution to the equation,
\begin{equation}
    \epsilon' + \log(g(\epsilon')) = \epsilon,
\end{equation} 
where
\begin{equation}
g(\epsilon') = 
\frac{1+\alpha'-(\alpha')^{d+1}-(\alpha')^{U-L+1-d}}
{1-(\alpha')^{U-L+1}}, 
\quad
d=\min\left\{\Delta, \left\lceil\frac{U-L}{2}\right\rceil\right\},
\quad
\alpha' = e^{-\epsilon' / \Delta}.
\end{equation}
The renormalized geometric mechanism outputs $\wt{M}_{\mathrm{RGM}}=M(D)+\delta$, where $\delta\in\mathbb{Z}$ is a draw from the truncated double-geometric distribution, 
\begin{equation}\label{eq:renormGeom}
    \Pr(\delta = k) = 
    \frac{e^{-|k| \epsilon' / \Delta }}
    {\sum_{l=L-M}^{U-M}e^{- |l|\epsilon' / \Delta }},
    \quad \text{for } k = L-M, L-M+1, \dots, U-M,
\end{equation}
where we use $M$ to stand for $M(D)$.
\end{definition}

\begin{theorem}
    The renormalized geometric mechanism satisfies $\epsilon$-DP.
\end{theorem}
\begin{proof}
    For any neighboring datasets $D$ and $D'$, denote $M(D')$ by $M'$ and $\wt{M}_{\mathrm{RGM}}(D')$ by $\wt{M}_{\mathrm{RGM}}'$. It suffices to show that $\Pr(\wt{M}_{\mathrm{RGM}} = s)\leq e^{\epsilon}\Pr(\wt{M}_{\mathrm{RGM}}' = s)$ for any $s\in\{L, L+1, \dots, U\}$. By definition of global sensitivity, $|M-M'|\leq \Delta$. Hence,
    
    \begin{align}
        \frac{\Pr(\wt{M}_{\mathrm{RGM}} = s)}{\Pr(\wt{M}_{\mathrm{RGM}}' = s)} 
        &= 
        \frac{e^{-\frac{\epsilon'}{\Delta}  |s-M|}}{e^{-\frac{\epsilon'}{\Delta}  |s-M'|}}
        \frac{\sum_{l=L}^{U}e^{- \frac{\epsilon'}{\Delta}  |l-M'|}}{\sum_{l=L}^{U}e^{-\frac{\epsilon'}{\Delta}  |l-M|}}\\
        &=
        e^{-\frac{\epsilon'}{\Delta}(|s-M| - |s-M'|)} 
        \left(
        \frac{1+\alpha'-(\alpha')^{M'-L+1}-(\alpha')^{U+1-M'}}
        {1+\alpha'-(\alpha')^{M-L+1}-(\alpha')^{U+1-M}}
        \right)\\
        &\leq
        e^{\epsilon' \frac{|M-M'|}{\Delta}} 
        \left(
        \frac{1+\alpha'-(\alpha')^{M'-L+1}-(\alpha')^{U+1-M'}}
        {1+\alpha'-(\alpha')^{M-L+1}-(\alpha')^{U+1-M}}
        \right)\\
        &\leq
        e^{\epsilon'} 
        g(\epsilon') =
        e^{\epsilon'+\log g(\epsilon')} =e^\epsilon.
    \end{align}
    
Here, the inequality in (S.1.7) holds because 
$$
\max_{\substack{M,M' \in [L,U]\\ |M-M'|\le \Delta}}
\frac{1+\alpha'-(\alpha')^{M'-L+1}-(\alpha')^{U+1-M'}}
        {1+\alpha'-(\alpha')^{M-L+1}-(\alpha')^{U+1-M}}
= \frac{1+\alpha'-(\alpha')^{d+1}-(\alpha')^{U-L+1-d}}
{1-(\alpha')^{U-L+1}}
= g(\epsilon').
$$
\end{proof}

\begin{figure}

\centering{

\includegraphics[width=\linewidth]{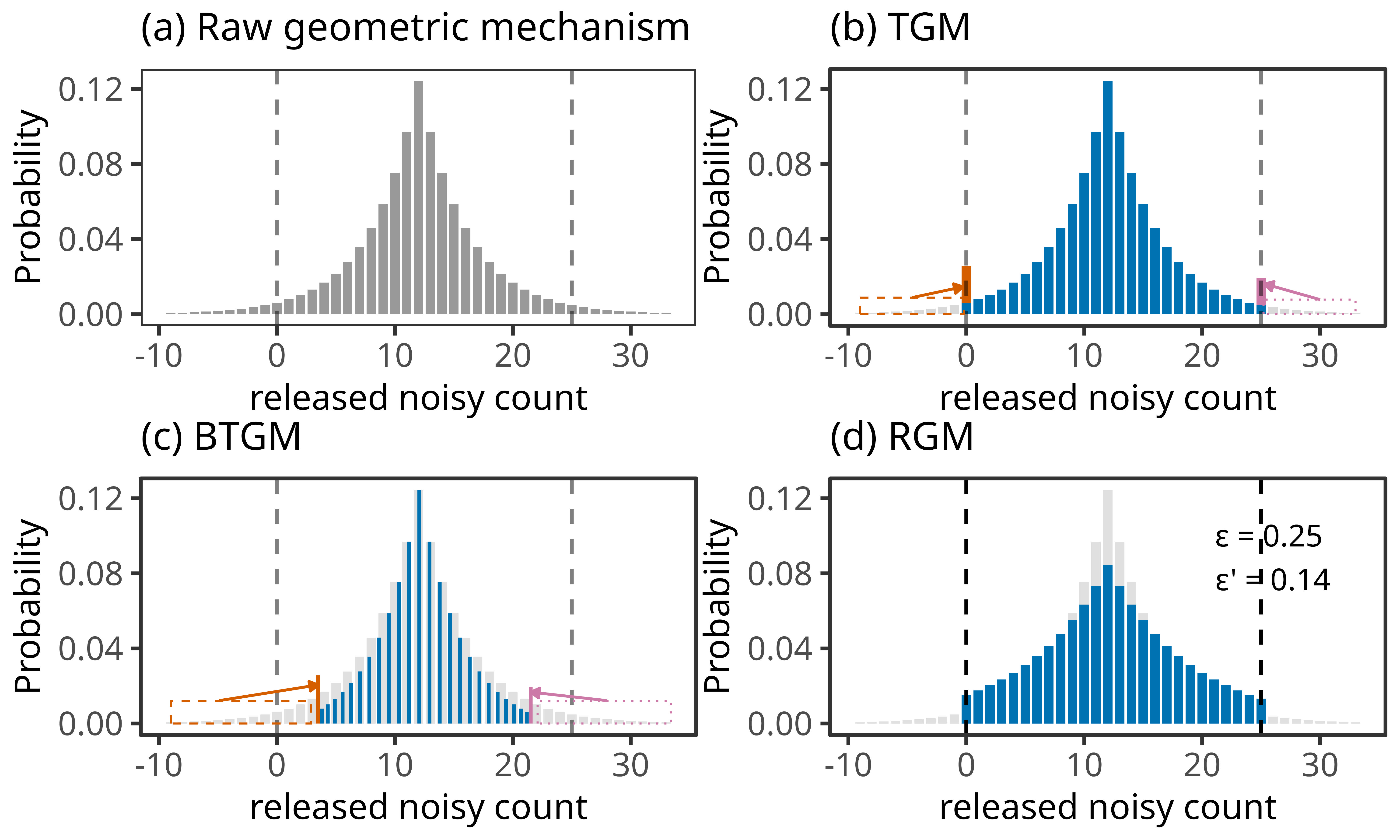}

}

\caption{\label{fig:mechanism}Comparison of three range-preserving mechanisms for $L=0,U=25,M=12,\epsilon=0.25$, and $\Delta=1$. Vertical dashed lines indicate the boundaries of the feasible range $[0,25]$. The light-gray bars show the probability distribution of the outputs from the raw geometric mechanism, while the blue foreground bars show those of the three truncated geometric mechanisms. Orange dashed rectangles with arrows (left) and pink dotted rectangles with arrows (right) indicate how out-of-range mass on each side is remapped back into $[0,25]$. The annotation $\epsilon' = 0.14$ in panel~(d) is defined as the solution to $\epsilon' + \log g(\epsilon') = \epsilon$ in the RGM definition above.}

\end{figure}

Figure \ref{fig:mechanism} illustrates the three mechanisms RGM, BTGM, and TGM. It also sheds some light on which of these three mechanisms might be preferred for particular problem settings. TGM and BTGM remap the out-of-range mass to two within-range points while leaving the other probabilities unchanged. BTGM additionally applies a slight inward contraction to all within-range values.
RGM, in contrast, discards mass outside $[L,U]$ and renormalizes the remaining distribution. Among the three, BTGM produces the most concentrated distribution; however, it does not guarantee integer-valued outputs. Unlike TGM and BTGM, RGM does not exhibit large boundary spikes that can cause distortions and overly concentrated outputs when $\epsilon$ is small. However, RGM has a relatively large and undesirable variance for modest $\epsilon$. When $\epsilon$ is set so that the geometric mechanism itself (with no truncation) seldom generates out-of-range values, all three mechanisms yield nearly identical distributions.

\section{Algorithms and Implementation Details}
\label{sec:algo}

This section presents the algorithms for the two methods presented in the main text, Bayes-NA and MLE-NN. Section~\ref{sec:algo_privacy} introduces the algorithm used to add noise to the counts to ensure differential privacy. Section~\ref{sec:algo_BayesNA} and Section~\ref{sec:algo_MLENN} detail the algorithms for Bayes-NA and MLE-NN, respectively, which take the noisy counts and turn them into estimates for the copula correlations.

\subsection{Algorithm to Ensure Differential Privacy}\label{sec:algo_privacy}

This subsection provides the algorithmic details corresponding to Section~3.2 of the main text.

\begin{algorithm}[H]
\caption{Calculate Differentially Private Statistics}
\label{alg:DPstatistic}
\textbf{Input:} Data matrix $D = \{x_{ij}\} \in \mathbb{R}^{n \times p}$; total privacy budget $\epsilon$\\
\textbf{Output:} Noisy co-occurrence statistics $\{\widetilde{t}_{jj'}: 1 \le j < j' \le p\}$

\begin{enumerate}
    \item Compute $t_{jj'} = \sum_{i=1}^n \mathbb{I}(x_{ij} \geq \operatorname{med}(X_j), x_{ij'} \geq \operatorname{med}(X_{j'}))$ for each $1 \le j < j' \le p$.

    \item For each $t_{jj'}$, apply a differentially private mechanism with sensitivity 1 and pair-wise privacy budget of $\frac{2\epsilon}{p(p-1)}$, to draw $\widetilde{t}_{jj'}$:
    \begin{enumerate}
        \item If Bayes-NA is implemented, use geometric mechanism.
        \item If MLE-NN is implemented, apply a range-preserving mechanism to ensure $\widetilde{t}_{jj'}\in[0,\frac{n}{2}]$. Choices include TGM, BTGM, or RGM.
    \end{enumerate}
\end{enumerate}
\end{algorithm}

\subsection{Algorithm for Bayesian Estimation of Copula Correlation}
\label{sec:algo_BayesNA}

This subsection provides the algorithmic details corresponding to Section~3.3 of the main text.

\begin{algorithm}[H]
\caption{Estimate $R$ using Bayesian noise-aware method (Bayes-NA)}
\label{alg:Bayes}
\textbf{Input:} $\{\widetilde{t}_{jj'}\in\mathbb{Z}: 1 \le j < j' \le p\}$ output from Algorithm~\ref{alg:DPstatistic}; significance level $\alpha\in(0,1)$.\\
\textbf{Output:} Differentially private correlation matrix $\widetilde{R}_{\mathrm{Bayes}}$ and element-wise $(1-\alpha)$-credible intervals.

\begin{enumerate}
    \item Draw posterior samples $\{R^{(s)}: s=1, \dots, S\}$ from the posterior distribution,
    \begin{equation*}
    p(R\mid \widetilde{T}_{12}=\widetilde{t}_{12},\dots,\widetilde{T}_{p-1,p}=\widetilde{t}_{p-1,p}) 
    \propto 
    \pi(R)
    \prod_{1\leq j<j'\leq p}
    \Pr(\widetilde{T}_{jj'}=\tilde t_{jj'}\mid R_{jj'}),
    \end{equation*}
    using \textsf{Stan}, where the prior $\pi(R)$ is chosen as LKJ(1).
    
    \item Compute the point estimate, 
    \begin{equation*}
        \widetilde{R}_{\mathrm{Bayes}}=\frac{1}{S}\sum_{s=1}^S R^{(s)},
    \end{equation*}
    and element-wise $(1-\alpha)$ credible interval $[\tilde q_{\alpha/2}^{jj'}, \tilde q_{1-\alpha/2}^{jj'}]$, where $\tilde q_{\alpha/2}^{jj'}$ and $\tilde q_{1-\alpha/2}^{jj'}$ denote the empirical $\alpha/2$ and $(1-\alpha/2)$ quantiles of $\{R_{jj'}^{(s)}: s=1, \dots, S\}$, with $R_{jj'}^{(s)}$ being the $(j,j')$-th element of $R^{(s)}$.
\end{enumerate}
\end{algorithm}

\subsection{Algorithm for MLE of Copula Correlation}
\label{sec:algo_MLENN}

This subsection provides the algorithmic details corresponding to Section~4 of the main text.

\begin{algorithm}[H]
\caption{Estimate $R$ using MLE noise-naive method (MLE-NN)}
\label{alg:MLE}
\textbf{Input:} $\{\widetilde{t}_{jj'}\in[0,\frac{n}{2}]: 1 \le j < j' \le p\}$ output from Algorithm~\ref{alg:DPstatistic}. \\
\textbf{Output:} Differentially private correlation matrix $\widetilde{R}_{\mathrm{MLE}}$.

\begin{enumerate}
    \item For each $1 \le j < j' \le p$, compute $\hat R_{jj'}$ by solving 
    \begin{equation*}
\label{eq:MLE}
   \widetilde{t}_{jj'} =  \left(\sum_{t=0}^{n/2} \binom{n/2}{t}^2  \left( \frac{\pi + 2 \arcsin(\hat{R}_{jj'})}{\pi - 2 \arcsin(\hat{R}_{jj'})} \right)^{2t}\right)^{-1} \left(\sum_{t=0}^{n/2} t  \binom{n/2}{t}^2  \left( \frac{\pi + 2 \arcsin(\hat{R}_{jj'})}{\pi - 2 \arcsin(\hat{R}_{jj'})} \right)^{2t}\right)
\end{equation*}

    using a root-finding algorithm such as bisection.
    
    \item Let $\widehat{R}=(\hat R_{jj'})_{1\leq j,j'\leq p}$. Use Higham's algorithm to solve 
        \begin{equation*}
            \widetilde{R}_{\mathrm{MLE}} = 
            \arg\min_{R \in \mathcal{C}_p}
            \bigl\| R - \widehat{R} \bigr\|_F^2,
            \text{ where }
            \mathcal{C}_p = \bigl\{ R \in \mathbb{R}^{p \times p} :\ R = R^\top,\ R \succeq \mathbf{0},\ \operatorname{diag}(R) = 1 \bigr\}.
        \end{equation*}
    Note that if $\widehat{R}$ is PSD, $\widetilde{R}_{\mathrm{MLE}}=\widehat{R}$.
\end{enumerate}
\end{algorithm}

\section{Additional Simulation Results}\label{sec:moresims}

This section includes additional simulation results.   Section \ref{sec:discrete} provides simulation results when using marginal distributions with many ties, i.e., discrete marginals.  Other details of the simulation settings correspond to those in Sections 5.1 and 5.2 of the main text. Section~\ref{sec:coverage_binwise} provides empirical coverage rates of the 95\% credible intervals presented in Section~5.3 of the main text. Section~\ref{sec:additional_mechanism} compares the MAEs when using MLE-NN with the three truncated geometric mechanisms described in Section \ref{sec:RGM}.   As in the simulations in the main text, for these simulation studies, all measurements are assessed across 1000 independent runs.  Section~ \ref{sec:cl_vs_joint} compares the posterior distributions computed with the composite likelihood and the full joint likelihood when $p=3$.  This dimension is  small enough that the full joint likelihood is computationally feasible to work with.

\subsection{Simulation Results for Discrete Marginals}
\label{sec:discrete}

To assess the effect of lexicographic tie-breaking introduced in Section~3.2 of the main text, we use simulations with $p=2$ discrete marginal distributions. The simulated data are generated from a Gaussian copula with 
\begin{equation*}
    R=\begin{pmatrix}1 & r \\ r & 1\end{pmatrix},
\end{equation*}
where $r\sim\operatorname{Unif}(-1,1)$. This is equivalent to the scaled Wishart distribution in the main text. One variable has an ordinal marginal distribution taking values in $\{1,2,3\}$ with probabilities of
$(0.3, 0.3, 0.4)$, respectively.  The second variable has  
an ordinal distribution taking values in $\{1,2,3,4,5,6\}$ with probabilities
    $(0.2, 0.1, 0.2, 0.2, 0.2, 0.1)$, respectively.

Figure~\ref{fig:MAE_discrete} displays the MAEs for the methods described in the main text when $\epsilon\in\{0.01,0.1,1\}$. 
Bayes-NA and MLE-NN attain the lowest MAEs, generally tracking each other closely. Even though the marginals are discrete with only a small number of levels, Bayes-NA and MLE-NN continue to perform well. These approaches rely only on binary comparisons with respect to the median and  therefore can be relatively insensitive to the number of discrete levels. In contrast, when the number of discrete levels is small, the Li-Kendall method tends to perform worse. As shown in Figure~\ref{fig:MAE_discrete}, for $\epsilon = 1$, even with large sample sizes, Li-Kendall does not reduce the MAE below 0.05, whereas Bayes-NA and MLE-NN achieve lower errors. 

As analyzed in the main text, the primary factor affecting Bayes-NA and MLE-NN is the number of ties at the median, since lexicographic tie-breaking for these observations may distort correlation estimation. However, as demonstrated in these simulations, when the proportion of median ties is not excessively large (e.g., 30\% in our setting), Bayes-NA and MLE-NN still can perform satisfactorily.

\begin{figure}[t]

\centering{

\includegraphics[width=\linewidth]{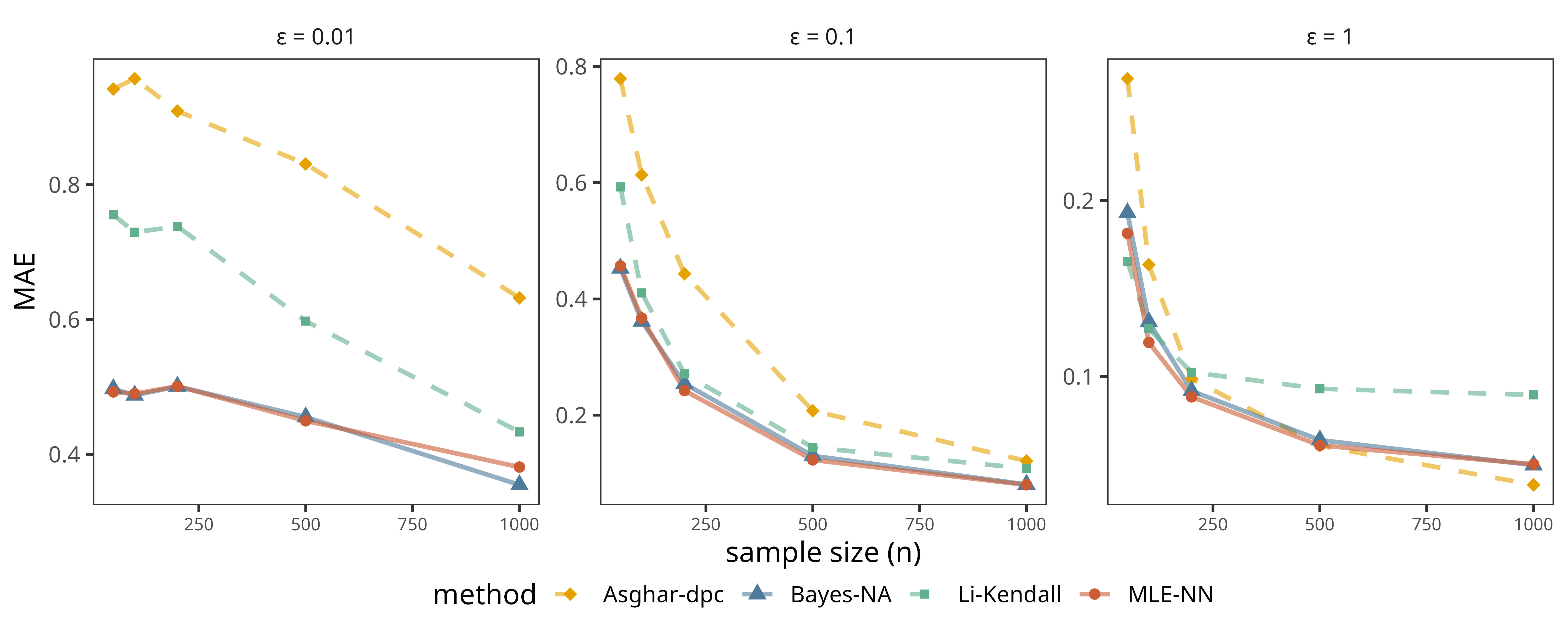}

}

\caption{\label{fig:MAE_discrete}MAEs of $\hat R_{12}$ when $p=2$ and  $n \in \{50, 100, 200, 500, 1000\}$ with varying $\epsilon$. Results based on 1000 simulation runs.  The maximum Monte Carlo error is 0.02. Solid lines: Bayes-NA, MLE-NN. Dashed lines: Li-Kendall, Asghar-dpc.}

\end{figure}

\subsection{Additional Detail on Coverage Rates when \texorpdfstring{$p=5$}{p=5}} 
\label{sec:coverage_binwise}

In Section 5.3 of the main text, we mention that Bayes-NA can experience lower than nominal coverage rates for $R_{jj'}$ near $\pm1$.   Figure~\ref{fig:coverage_pairwise} provides supporting evidence for this claim for the simulations where $p=5$.
We observe lower than nominal coverage rates for $R_{jj'}\approx\pm 1$ when both $n$ and $\epsilon$ are small. As discussed in the main text, when $\epsilon$ and $n$ are small, the posterior distribution is dominated by the prior distribution, and the positive semi-definite constraints force the LKJ prior for $p=5$ to assign little probability mass to correlations near $\pm 1$. Consequently, the posterior samples do not explore regions near the boundaries adequately. When either $n$ or $\epsilon$ becomes sufficiently large, the likelihood is the main contributor to the posterior distribution, and the credible intervals achieve satisfactory empirical coverage even for correlations near the boundaries. 

\begin{figure}[t]

\centering{

\includegraphics[width=\linewidth]{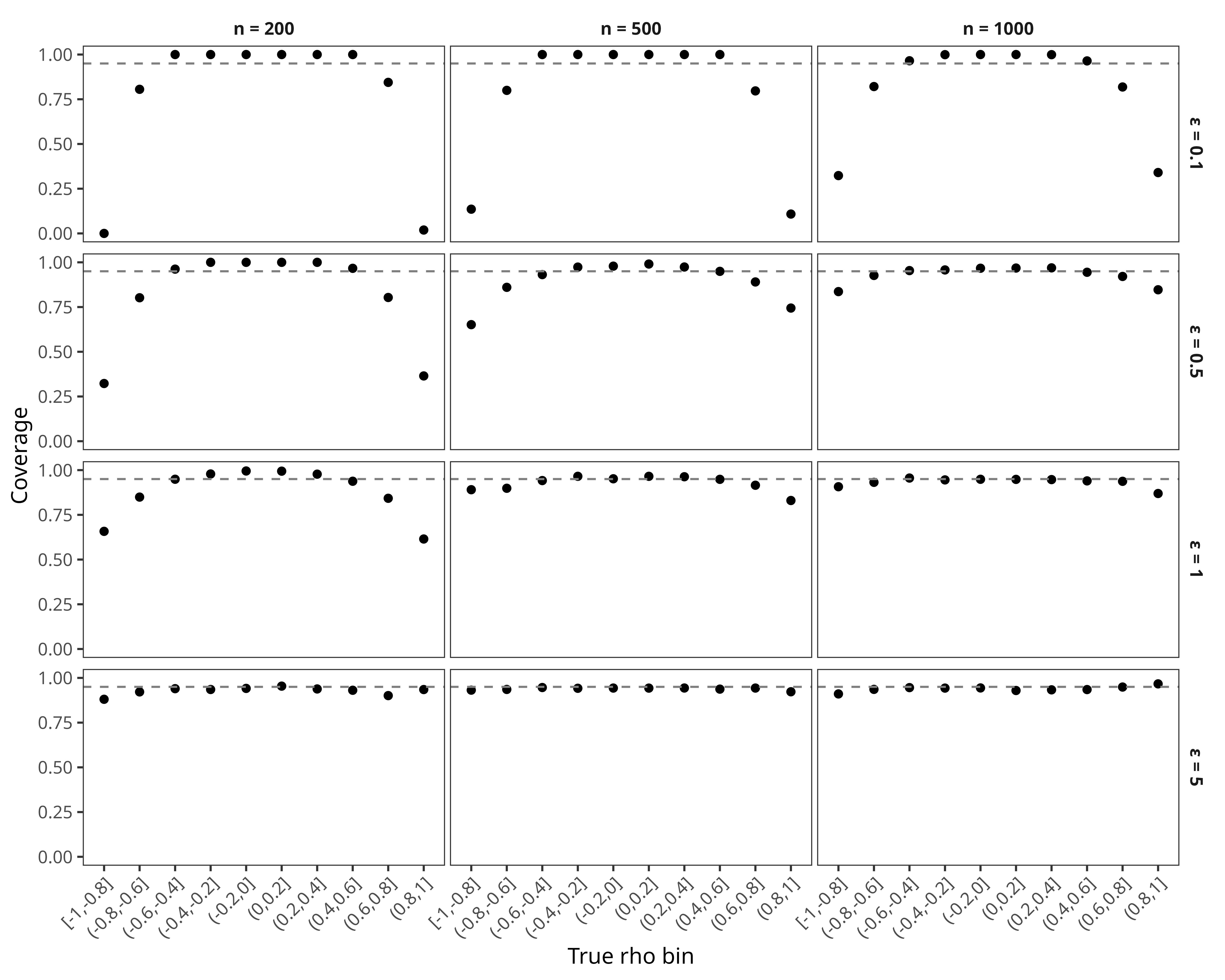}

}

\caption{\label{fig:coverage_pairwise}Empirical coverage rates of the 95\% credible intervals produced by Bayes-NA when $p=5$, computed within bins of $R_{jj'}$ of width 0.2. Panels correspond to different sample sizes $n$ and privacy budgets $\epsilon$. The dashed horizontal line indicates the nominal 95\% coverage level.}

\end{figure}

\subsection{Comparison of Different Truncated Geometric Mechanisms}
\label{sec:additional_mechanism}

In this section, we implement MLE-NN using each of the three truncated DP mechanisms, namely TGM and BTGM described in Section 4.2 of the main text and RGM described in Section~\ref{sec:RGM}.

Figure~\ref{fig:MAE_other} displays the MAEs for the three truncated geometric mechanisms, along with the MAEs of Bayes-NA. When both $\epsilon$ and $n$ are large, all methods converge to similarly small MAE values. In other cases, MLE-NN with BTGM tends to outperform MLE-NN with TGM or RGM for the settings considered in our simulations. All MLE-NN variants, regardless of the truncated geometric mechanism used, perform slightly worse than Bayes-NA. However, the differences in MAEs are small.

\begin{figure}

\centering{

\includegraphics[width=\linewidth]{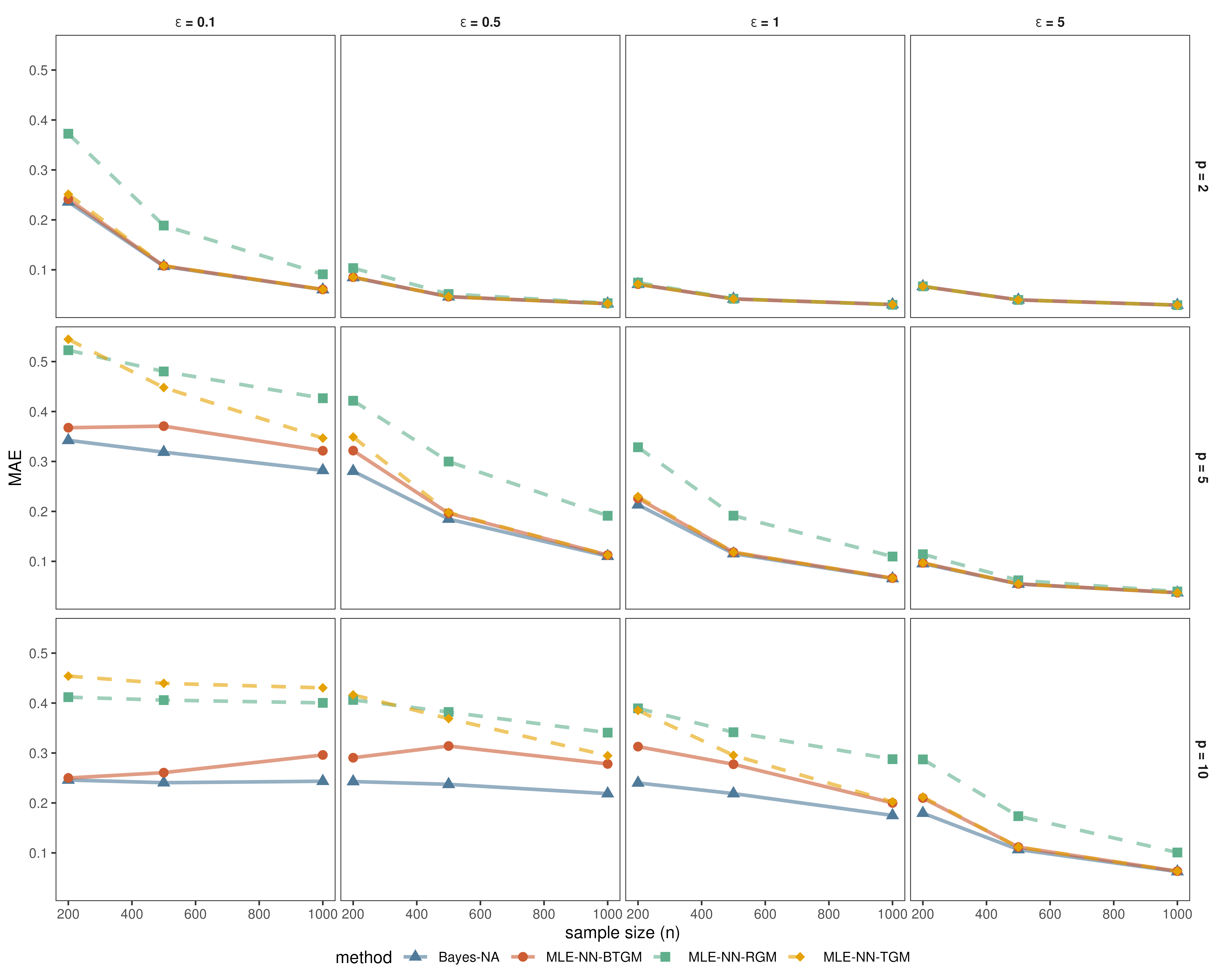}

}

\caption{\label{fig:MAE_other}MAEs of $\hat R$ for settings with $n \in \{200, 500, 1000\}$ with varying $\epsilon$ and $p$. Solid lines: Bayes-NA and MLE-NN with BTGM. Dashed lines: MLE-NN with RGM and MLE-NN with TGM.}

\end{figure}

\subsection{Comparison of Composite Likelihood and Full Joint Likelihood} \label{sec:cl_vs_joint}

The composite likelihood (CL) in (18) of the main text is used as an approximation for the intractable joint likelihood
$\Pr(\widetilde T_{12}, \dots, \widetilde T_{p-1, p} \mid R)$. 
Here we assess the accuracy of this approximation empirically for $p=3$, which is small enough dimension that the joint probability mass function is computationally tractable. After
conditioning on the column totals, the total-sample constraint, and the three
pairwise constraints, only one inner cell count is unspecified, so the joint
probability mass function reduces to a one-dimensional sum that can be evaluated in closed form. Specifically, we have

\begin{equation}\label{eq:joint_p3}
\begin{aligned}
&\Pr(\widetilde{T}_{12} = \widetilde{t}_{12},\,\widetilde{T}_{13} = \widetilde{t}_{13},\,\widetilde{T}_{23} = \widetilde{t}_{23} \mid R) \\
&\quad = \sum_{t_{12},\, t_{13},\, t_{23}}
   \Pr(\widetilde{T}_{12} = \widetilde{t}_{12},\,\widetilde{T}_{13} = \widetilde{t}_{13},\,\widetilde{T}_{23} = \widetilde{t}_{23} \mid T_{12} = t_{12},\, T_{13} = t_{13},\, T_{23} = t_{23})\\
&\qquad\quad \times \Pr(T_{12} = t_{12},\, T_{13} = t_{13},\, T_{23} = t_{23} \mid R) \\
&\quad = \sum_{t_{12},\, t_{13},\, t_{23}}
   \prod_{1 \le j < j' \le 3}
     \Pr(\widetilde{T}_{jj'} = \widetilde{t}_{jj'} \mid T_{jj'} = t_{jj'})\\
&\qquad\quad \times \Pr(T_{12} = t_{12},\, T_{13} = t_{13},\, T_{23} = t_{23} \mid R) \\
&\quad = \sum_{t_{12},\, t_{13},\, t_{23}}
   \prod_{1 \le j < j' \le 3}
     \Pr(\widetilde{T}_{jj'} = \widetilde{t}_{jj'} \mid T_{jj'} = t_{jj'})\\
&\qquad\quad \times
   \sum_{k = k_{\min}(t)}^{k_{\max}(t)}
     \Pr(T_{12} = t_{12},\, T_{13} = t_{13},\, T_{23} = t_{23},\, n_{111} = k \mid R),
\end{aligned}
\end{equation}
where $\Pr(\widetilde{T}_{jj'} = \widetilde{t}_{jj'} \mid T_{jj'} = t_{jj'})$ is the double-geometric probability mass function in (3) of the main text, $k = n_{111} = \sum_{i=1}^{n} \mathbb{I}(z_{i1} \geq 0, z_{i2} \geq 0, z_{i3} \geq 0)$ is the single free inner cell count, and $k_{\min}(t)$, $k_{\max}(t)$ are the cell-non-negativity bounds on $k$. The last probability $\Pr(T_{12} = t_{12},\, T_{13} = t_{13},\, T_{23} = t_{23},\, n_{111} = k \mid R)$ can be derived similarly as (10) of the main text by conditioning the multinomial distribution of the eight $2\times 2\times 2$ cell counts on the column totals being $n/2$, with the trivariate cell probabilities obtained by an argument analogous to Lemma~\ref{lem:orthant}.

Given the closed form in \eqref{eq:joint_p3}, we obtain the joint posterior by importance-sampling reweighting of the composite likelihood Bayes-NA posterior on the same noisy dataset. Each draw $R^{(i)}$ from the proposal receives weight
\begin{equation*}
w_{i} \;\propto\; \frac{\Pr(\widetilde{\mathcal{T}} = \widetilde{t} \mid R^{(i)})}{f_{\mathrm{CL}}(\widetilde{\mathcal{T}} \mid R^{(i)})},
\end{equation*}
where the numerator is the right-hand side of \eqref{eq:joint_p3} with $\widetilde{\mathcal{T}} = (\widetilde{T}_{12}, \widetilde{T}_{13}, \widetilde{T}_{23})$ and the denominator is (18) of the main text. Posterior summaries under the joint likelihood are then obtained by standard self-normalized importance-sampling estimators using these weights.

We evaluate both estimators on the grid $n \in \{50, 100, 200\}$ and $\epsilon \in \{0.1, 1\}$ with $R$ fixed at $R_{jj'} = 0.7$ for all $j \neq j'$, a setting where the CL approximation is most likely to differ from the joint. For each $(n, \epsilon)$ cell we run $100$ replicates. 
Figure~\ref{fig:estCLvsJoint} displays the two posteriors for a single replicate for each combination of $(n, \epsilon)$.  All posterior means coincide and the $95\%$ credible intervals overlap almost exactly. This result suggests that the composite-likelihood approximation in this example offers accurate marginal point and interval estimates.
We also repeat the simulation for $100$ replicates of each scenario.  As seen in Figure~\ref{fig:mae_CLvsJoint}, the average MAE of each $R_{jj'}$ is  nearly identical between the CL posterior and the full joint posterior.

While marginally the CL approximation and full posterior offer similar results, the CL approximation does introduce  some differences in the multivariate posterior distribution.   Figure~\ref{fig:corrCLvsJoint} displays the within-replicate posterior correlation between pairs of entries of $R$, averaged over the $100$ replicates. There appears to be some attenuation in the CL-estimated associations across correlations. This attenuation results from the independence assumption. 
The gap is more pronounced at $\epsilon = 1$ than at $\epsilon = 0.1$, since the larger independent geometric noise added at small $\epsilon$ partially masks the underlying cross-pair dependence and brings the noisy pairwise statistics closer to the independence assumption that CL imposes.

While there is attenuation, 
we anticipate that many analysts are likely to focus on marginal point and interval estimates for $R_{jj'}$ rather than the multivariate distribution.  An analogous situation occurs when summarizing inferences for regression coefficients: analysts mainly use marginal inferences (e.g., point estimates and standard errors) when interpreting results.  
Nonetheless, finding estimation routines to better approximate the multivariate structure is a useful area for future research.

\begin{figure}[t]

\centering{

\includegraphics[width=\linewidth]{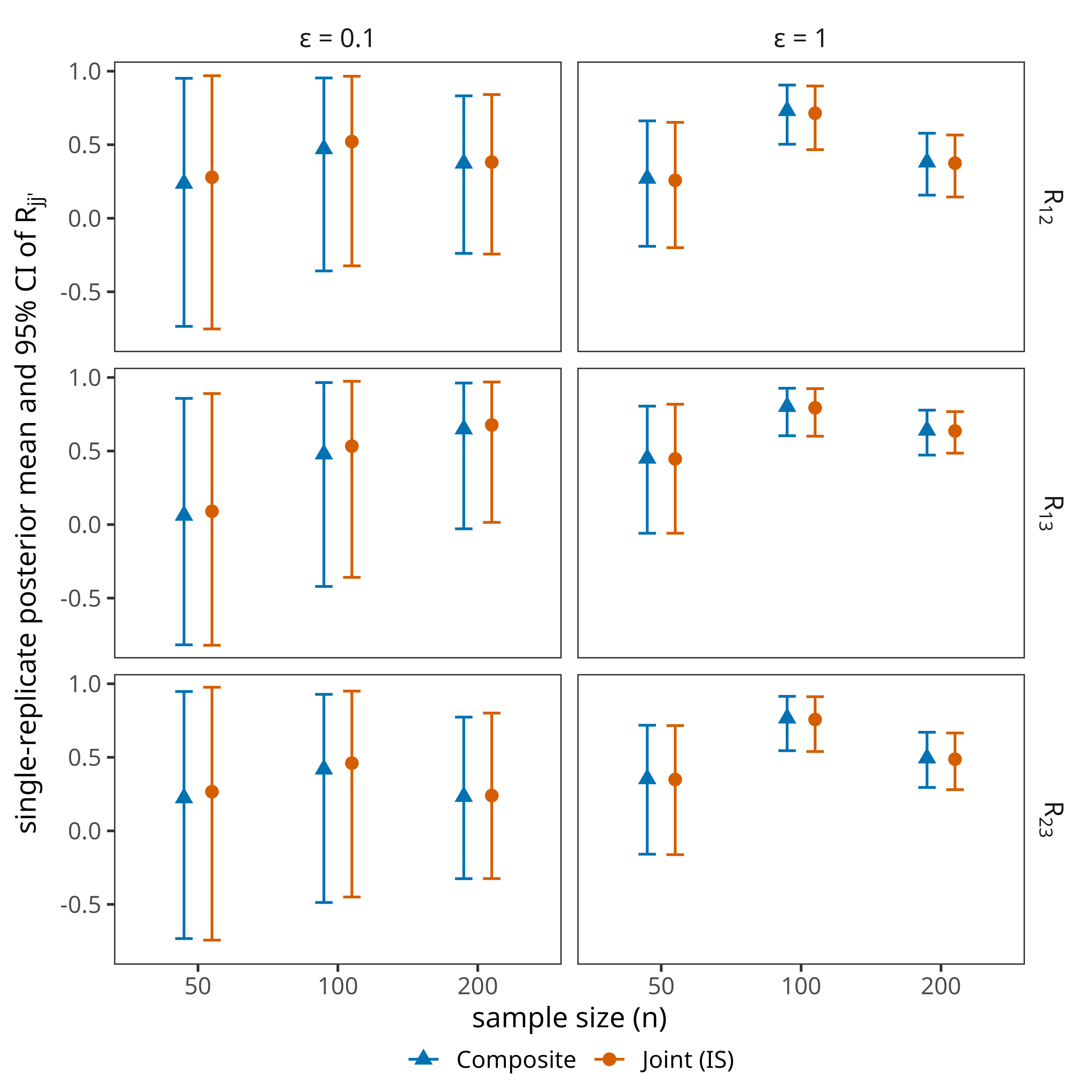}

}

\caption{\label{fig:estCLvsJoint}Posterior mean and $95\%$ credible interval for each pairwise correlation $R_{jj'}$ from a single replicate, under the composite-likelihood Bayes-NA estimator (Composite; blue triangles) and the full-joint-likelihood Bayes-NA estimator (Joint (IS); orange circles). Rows correspond to the three pairs $R_{12}, R_{13}, R_{23}$ and columns to the privacy budget $\epsilon \in \{0.1, 1\}$. The true value is $R_{jj'} = 0.7$ for all $j \ne j'$.}

\end{figure}

\begin{figure}[t]

\centering{

\includegraphics[width=\linewidth]{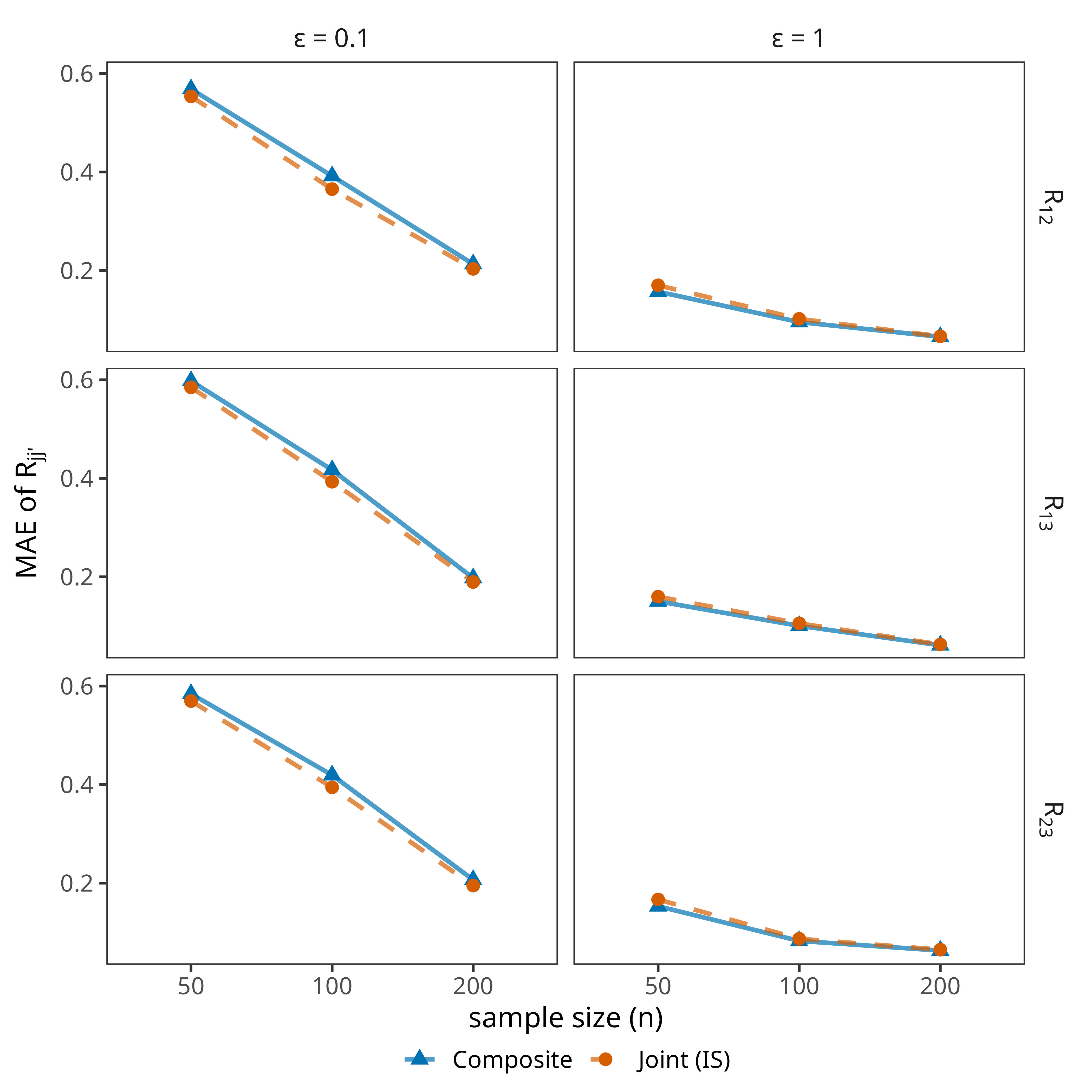}

}

\caption{\label{fig:mae_CLvsJoint}Posterior MAE of each $R_{jj'}$ under the composite-likelihood Bayes-NA estimator (Composite; blue solid line with triangles) and the full-joint-likelihood Bayes-NA estimator (Joint (IS); orange dashed line with circles). Rows correspond to the three pairs $R_{12}, R_{13}, R_{23}$ and columns to the privacy budget $\epsilon \in \{0.1, 1\}$.}

\end{figure}

\begin{figure}[t]

\centering{

\includegraphics[width=\linewidth]{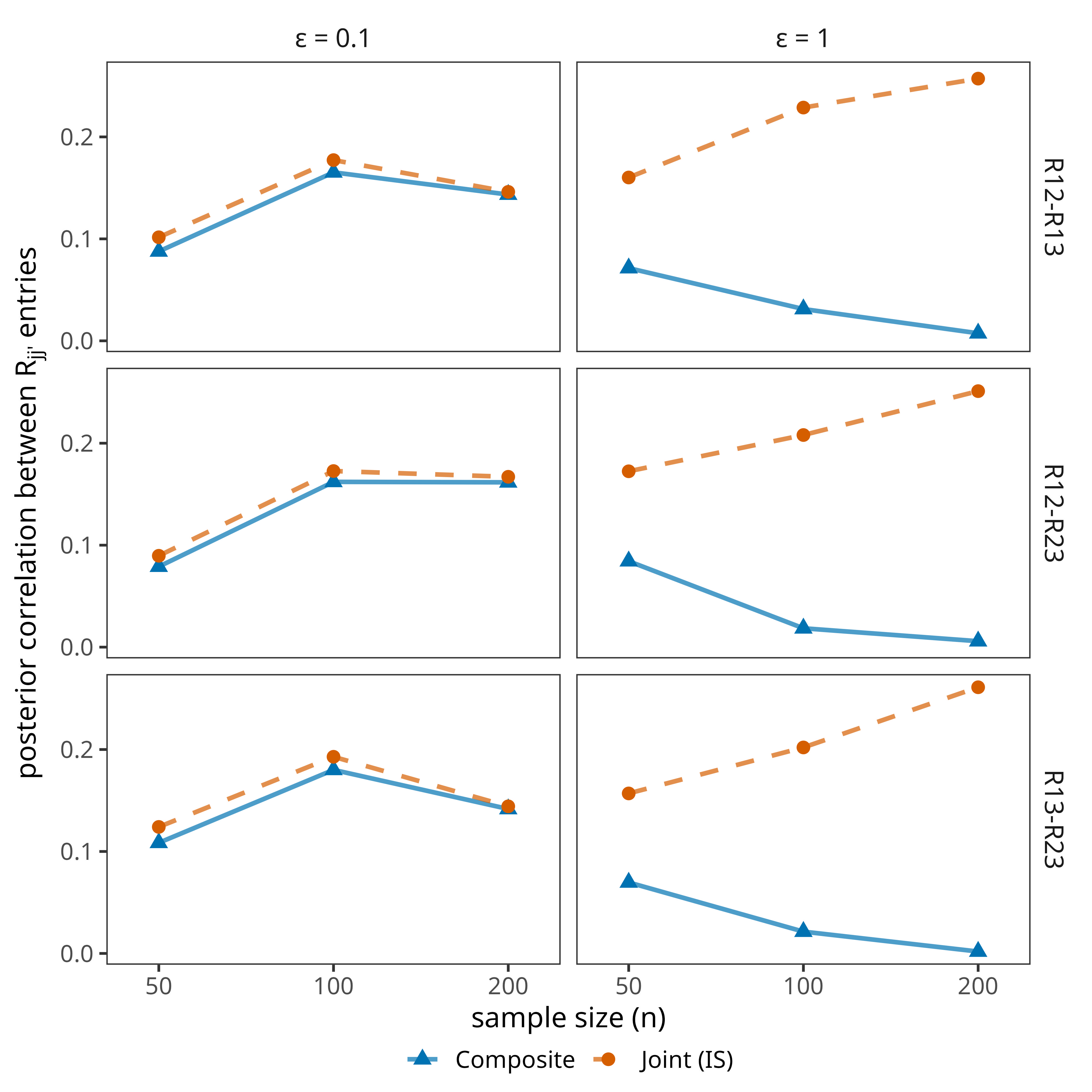}

}

\caption{\label{fig:corrCLvsJoint}Posterior correlation between pairs of entries of $R$ under the composite-likelihood Bayes-NA estimator (Composite; blue solid line with triangles) and the full-joint-likelihood Bayes-NA estimator (Joint (IS); orange dashed line with circles). Rows correspond to the three entry pairs $(R_{12}, R_{13})$, $(R_{12}, R_{23})$, $(R_{13}, R_{23})$ and columns to the privacy budget $\epsilon \in \{0.1, 1\}$.}

\end{figure}

\section{Additional Results for Analysis of Dietary Data} \label{sec:real_data}

Analysts may be interested in conditional relationships among variables.  These can be obtained from the posterior distributions of the copula correlation, as we now illustrate using the analysis of the dietary data from Section 6 of the main text.

Specifically, for each of the variables Grain, Dairy, Veg, PF, and Sugar, we make inferences for the slope of the regression of that variable on income, controlling for Kcal.  Let $j$ index one of the outcomes, e.g., Grain. In each posterior draw $s=1, \dots, S=1000$ of $R$, for variable $j$ we compute 
\begin{equation*}
\beta^{(s)}_{j}=\frac{R^{(s)}_{j,\text{Income}}-R^{(s)}_{j,\text{Kcal}} R^{(s)}_{\text{Income}, \text{Kcal}}}{1-(R^{(s)}_{\text{Income}, \text{Kcal}})^2}.
\end{equation*}
Figure~\ref{fig:coef_cond} displays the posterior mean, 2.5\% and 97.5\% quantiles based on $\{\beta^{(1)}_{j}, 
\dots, \beta^{(S)}_{j}\}$ for each variable $j$.
Conditional on Kcal, income is positively associated with the consumption of vegetables and protein foods, and negatively associated with added sugar, grain and dairy. For $\epsilon=1$, the posterior intervals are notably wider than the observed data intervals, typically overlapping zero  slightly. This is the trade-off in accuracy demanded for the stronger privacy guarantee. When $\epsilon =3$, the intervals are much narrower than when $\epsilon=1$. We conjecture that many analysts would reach very similar conclusions using the $3$-DP intervals and the non-private intervals.  Overall, an analyst may conclude from the results that, for example, individuals with higher income tend to consume calories more from vegetables and less from sugar-sweetened foods.

\begin{figure}[t]

\centering{

\includegraphics[width=0.7\linewidth,height=\textheight]{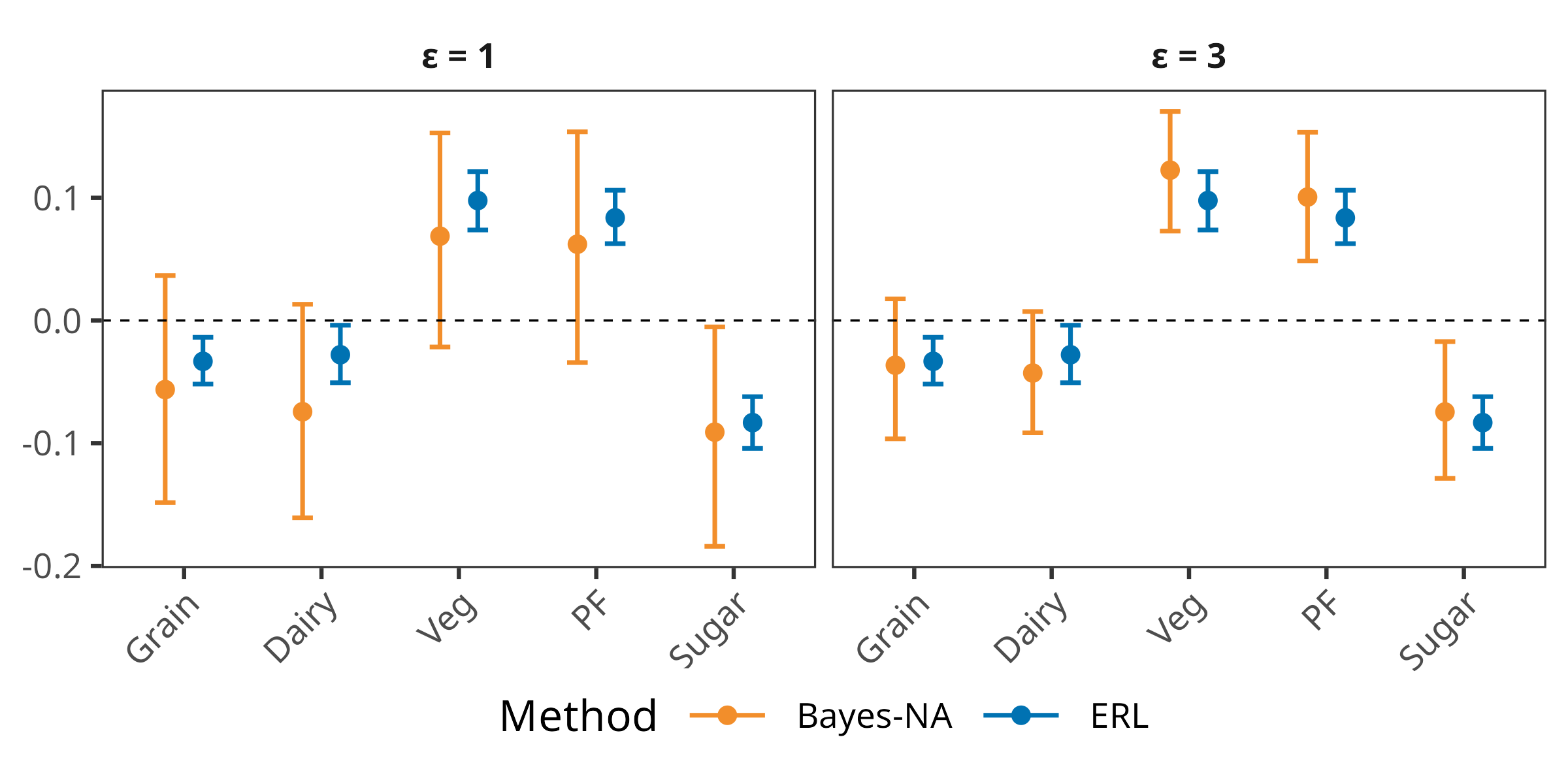}

}

\caption{\label{fig:coef_cond}Posterior summaries of regression coefficients $\beta_i$ for NHANES and FPED data. Within each outcome, the orange interval on the left is the Bayes-NA private estimate and the blue interval on the right is the non-private ERL estimate. Dot: posterior mean; bars: 2.5\% and 97.5\% posterior quantiles.}

\end{figure}

\section{Derivation of the Latent Cell Probabilities}\label{sec:orthant}

This section provides a self-contained derivation of the cell probabilities $p_{uv}(R_{jj'})$, $u,v\in\{0,1\}$, that appear in (7) and (8) of Section~3.1 of the main text.

\begin{lemma}\label{lem:orthant}
Let $(Z_j, Z_{j'})$ be standard bivariate normal with correlation $R_{jj'}\in(-1,1)$. Then
\begin{align}
p_{11}(R_{jj'}) &:= \Pr(Z_j \geq 0,\ Z_{j'} \geq 0)
   = \frac{1}{4} + \frac{1}{2\pi}\arcsin(R_{jj'}), \label{eq:orthant1}\\
p_{00}(R_{jj'}) &:= \Pr(Z_j < 0,\ Z_{j'} < 0) \;=\; p_{11}(R_{jj'}),\\
p_{10}(R_{jj'}) &:= \Pr(Z_j \geq 0,\ Z_{j'} < 0) \;=\; p_{01}(R_{jj'}) \;=\; \frac{1}{4} - \frac{1}{2\pi}\arcsin(R_{jj'}).
\end{align}
\end{lemma}

\begin{proof}
Let $U_1, U_2$ be independent $N(0,1)$ random variables, and define
\begin{equation*}
Z_j = U_1, \qquad Z_{j'} = R_{jj'}\,U_1 + \sqrt{1-R_{jj'}^2}\, U_2,
\end{equation*}
so that $(Z_j, Z_{j'})$ is standard bivariate normal with correlation $R_{jj'}$. Switching to polar coordinates $U_1 = \rho\cos\theta$, $U_2 = \rho\sin\theta$, the angle $\theta$ is uniformly distributed on $[0,2\pi)$ and is independent of the radius $\rho\geq 0$. The event $\{Z_j\geq 0,\ Z_{j'}\geq 0\}$ is therefore equivalent to
\begin{equation*}
\cos\theta \geq 0 \quad\text{and}\quad R_{jj'}\cos\theta + \sqrt{1-R_{jj'}^2}\,\sin\theta \geq 0.
\end{equation*}

Define $\psi := \arcsin(R_{jj'}) \in (-\pi/2, \pi/2)$, so that $\sin\psi = R_{jj'}$ and $\cos\psi = \sqrt{1 - R_{jj'}^2}$. Using the identity $\sin(\theta + \psi) = \sin\theta\cos\psi + \cos\theta\sin\psi$, the second inequality becomes
\begin{equation*}
R_{jj'}\cos\theta + \sqrt{1-R_{jj'}^2}\,\sin\theta = \sin\psi\,\cos\theta + \cos\psi\,\sin\theta = \sin(\theta + \psi) \geq 0.
\end{equation*}

Thus $Z_j\geq 0, Z_{j'}\geq 0\Longleftrightarrow\theta\in\left[-\psi,\ \tfrac{\pi}{2}\right]$, which has Lebesgue measure
\begin{equation*}
\tfrac{\pi}{2} - (-\psi) = \tfrac{\pi}{2} + \arcsin(R_{jj'}).
\end{equation*}

Since $\theta$ is uniformly distributed on $[0, 2\pi)$,
\begin{equation*}
\Pr(Z_j \geq 0,\ Z_{j'} \geq 0) = \frac{1}{2\pi}\!\left[\tfrac{\pi}{2} + \arcsin(R_{jj'})\right] = \tfrac{1}{4} + \tfrac{1}{2\pi}\arcsin(R_{jj'}),
\end{equation*}
which is \eqref{eq:orthant1}.

The remaining identities follow from the symmetry of the centered bivariate normal density. Because $(Z_j, Z_{j'}) \stackrel{d}{=} (-Z_j, -Z_{j'})$, we have $p_{00}(R_{jj'}) = p_{11}(R_{jj'})$ and $p_{10}(R_{jj'}) = p_{01}(R_{jj'})$. Combined with the constraint $\sum_{u,v\in\{0,1\}} p_{uv}(R_{jj'}) = 1$, this gives $p_{01}(R_{jj'}) = 1/2 - p_{11}(R_{jj'}) = 1/4 - \arcsin(R_{jj'})/(2\pi)$.
\end{proof}